\documentclass[lettersize,journal]{IEEEtran}
\usepackage{amsmath,amsfonts}
\usepackage{algorithmic}
\usepackage{algorithm}
\usepackage{array}
\usepackage[caption=false,font=normalsize,labelfont=sf,textfont=sf]{subfig}
\usepackage{textcomp}
\usepackage{stfloats}
\usepackage{url}
\usepackage{verbatim}
\usepackage{graphicx}
\usepackage{cite}
\usepackage{svg}
\usepackage{cuted}
\usepackage{amssymb}
\usepackage{algorithm}
\usepackage{algorithmic}
\usepackage{bm} 
\usepackage{amsthm,mathtools,mathrsfs}
\usepackage{braket}

\usepackage{makecell} 
\DeclareMathOperator*{\argmax}{arg\,max}
\usepackage{tikz}
\usetikzlibrary{positioning}
\usetikzlibrary{calc} 
\usetikzlibrary{shapes.geometric} 
\usepackage{orcidlink}

\usepackage{hhline}
\usepackage{multirow}

\theoremstyle{remark}
\newtheorem{theorem}{Theorem}
\newtheorem{corollary}{Corollary}

\begin{document}

\title{A Structured Family of Grassmannian Constellations via Geodesic Mapping for MIMO Noncoherent Communications}

\author{Álvaro Pendás-Recondo\orcidlink{0000-0002-8029-8024}, and Enrique Pendás-Recondo\orcidlink{0000-0002-0790-7372}
\thanks{The work of Álvaro Pendás-Recondo was supported in part by Ministerio de Ciencia e Innovación and Agencia Estatal de Investigación (MICIN)/(AEI)/10.13039/501100011033 under projects PID2023-146246OB-C32 (ANT4IT) and PID2024-160783OA-I00 (RARIN-6G); and in part by the  Gobierno del Principado de Asturias under grants Asturias-Sekuens/FEDER IDE-2024-000693, and “Severo Ochoa” Program Grant PA-22-BP21-116. The work of Enrique Pendás-Recondo was supported in part by Ministerio de Ciencia e Innovación and Agencia Estatal de Investigación (MICIN)/(AEI)/10.13039/501100011033 under project  PID2021-124157NB-I00, cofunded by “ERDF A way of making Europe”; and in part by Fundación Séneca-Agencia de Ciencia y Tecnología de la Región de Murcia, under project “Ayudas a proyectos para el desarrollo de investigación científica y técnica por grupos competitivos (Comunidad Autónoma de la Región de Murcia)”, included in “Programa Regional de Fomento de la Investigación Científica y Técnica (Plan de Actuación 2022)”, REF. 21899/PI/22. 

Álvaro Pendás-Recondo is with the Signal Theory and Communications Group, Department of Electrical Engineering, Universidad de Oviedo, 33203, Gijón, Spain (e-mail: pendasalvaro@uniovi.es). Enrique Pendás-Recondo is with the Algebra Group, Department of Mathematics, Universidad de Oviedo, 33007, Oviedo, Spain (e-mail: pendasenrique@uniovi.es).
}}



\maketitle

\begin{abstract}
This work presents a novel structured family of Grassmannian constellations for multiple-input multiple-output (MIMO) noncoherent communications over Rayleigh block-fading channels, where neither the transmitter nor the receiver has channel state information (CSI). The proposed constellation design is built upon the geodesic curves of the Grassmann manifold, thereby exploiting its underlying geometric structure. The resulting solution is limited in spectral efficiency (with a maximum constellation size of $4M^2$ points, where $M$ is the number of transmit antennas), targeting a rate in the range of $0.25$–$1$ bps/Hz. However, all space-time matrices resulting from this design exhibit the remarkable property of having a single nonzero entry per row, meaning that only one transmit antenna is active per time slot. This property significantly reduces hardware complexity and implementation cost, while also lowering power consumption, as only a single radio-frequency (RF) chain and power amplifier are required for transmission. Furthermore, within the constellation size limits, the proposed design achieves error performance comparable to state-of-the-art optimization-based unstructured designs, as validated through symbol error rate (SER) numerical results. It also enables simple yet effective bit labeling, confirmed by comparisons of bit error rate (BER) and SER, and reduces the computational complexity of the maximum-likelihood (ML) detector for Grassmannian constellations by a factor of $M$.

\end{abstract}


\section{Introduction}
\IEEEPARstart{M}{ultiple}-input multiple-output (MIMO) communications have had a profound impact on wireless communications in recent decades, thanks to their enhanced spectral efficiency compared to single-input single-output (SISO) systems. This advantage has been extensively demonstrated not only mathematically, in terms of system capacity, but also in practice through their widespread deployment in modern wireless standards, albeit with associated challenges \cite{mimo1}.

\subsection{Motivation: Coherent versus Noncoherent Schemes and Unitary Space-Time Modulation}

One of the most significant challenges of MIMO is acquiring accurate channel state information (CSI). Channel state information at the transmitter (CSIT) in MIMO or multiple-input single-output (MISO) systems enables the application of digital precoding to achieve enhanced signal-to-noise ratio (SNR), diversity gain, or, only in MIMO, multiplexing gain.  Channel state information at the receiver (CSIR) is required in standard MIMO, MISO, single-input multiple-output (SIMO), and SISO systems to ensure correct symbol decoding. CSIR in multi-antenna systems is typically obtained from the observation of orthogonal pilot symbols sent periodically by the transmitter, whereas CSIT is acquired through feedback from the receiver.

Although accurate CSIR is generally feasible to obtain, its main drawback lies in the reduction of spectral efficiency caused by pilot symbol transmission. This limitation becomes pronounced when the number of antennas is large, even in slowly varying channels, or when the channel varies rapidly, even with relatively few antennas \cite{mimo_tran}. Regarding CSIT, its accurate acquisition is considerably more challenging, as it is affected by factors that do not compromise CSIR, including limited and quantized feedback and feedback delays that lead to outdated CSIT, among others \cite{imcsit}, \cite{ojom}. Moreover, the feedback channel always consumes additional power and spectral resources. 

In this context, noncoherent schemes, where neither the transmitter nor the receiver has CSI, are particularly attractive for short-packet communications with low-latency requirements and high user mobility, where conventional channel estimation is costly or unreliable, as well as for systems based on short bursts transmissions. While not adopted in current wireless standards, they represent a promising research direction with potential applications in paradigms such as ultra-reliable low-latency communications (URLLC),  Internet of Things (IoT) networks with significant energy and bandwidth constraints, and massive MIMO systems operating in rapidly time-varying channels \cite{thesis}, \cite{nonch_review}.

The problem of noncoherent MIMO was first considered in \cite{gt1} under Rayleigh block-fading. In this model, the channel entries are independent and identically distributed (i.i.d.) and remain constant during the transmission of several symbols before adopting an independent realization drawn from a normal distribution. In \cite{gt2}, \cite{gt2_2}, and \cite{gt3}, further studies on the capacity of this scenario were conducted. Remarkably, noncoherent MIMO systems can achieve a significant fraction of the coherent capacity gain at high SNR, and this fraction increases with the number of time symbols over which the channel remains constant \cite{gt3}. Furthermore, the use of unitary space-time modulation is an optimal solution in this scenario, in the sense that it minimizes the union bound of error probability under the assumption of equal-energy signals \cite{gt2_2}. Consequently, this approach to the problem has been extensively studied \cite{gt2}, \cite{gt2_2}, \cite{gt3}, \cite{gt4}. Mathematically, unitary space-time modulation is based on transmitting space-time matrices that are Stiefel representatives of points in the complex Grassmann manifold \cite{ghb}, which is why this solution is commonly referred to as Grassmannian signaling.

\subsection{Related Works: Constellation Designs}

For noncoherent MIMO communications, the problem of designing a Grassmannian constellation consists of finding a set of matrices that minimizes symbol error probability (SEP), with energy and power normalization inherently guaranteed by the unitary space-time condition. For the asymptotic case with an infinite number of points, the optimal design approach is to randomly sample the Grassmann manifold. However, in practical scenarios with a finite number of constellation points, the problem becomes considerably more complex. The packing problem on the Grassmann manifold endowed with a Riemannian metric---which induces the natural notion of distance on the manifold, also known as geodesic distance---remains unsolved for an arbitrary number of points \cite{packgrass}. Furthermore, beyond the geodesic distance, other figures of merit have been investigated for minimizing SEP or pairwise error probability (PEP), which, depending on the scenario, may serve as more effective predictors. Examples include the chordal distance \cite{do1}, \cite{grassma}, the diversity product (DP) \cite{do2}, \cite{do3}, also known as the coherence criterion, and the asymptotic union bound (UB) \cite{ubp}, \cite{ub}.

In general, constellation designs in the literature can be classified into two categories: structured and unstructured \cite{thesis}, \cite{ub}. The first category imposes a certain structure on the space-time matrices, which simplifies both the design process and symbol decoding \cite{syste_di}, \cite{expmap1}, \cite{expmap2}, \cite{geo_di}, \cite{systema}, \cite{struct_can}. This is a considerable advantage, since the computational complexity of the maximum-likelihood (ML) detector for arbitrary Grassmannian constellations represents a major challenge for practical implementation as the number of points (and thus the spectral efficiency) increases. In particular, the constellations presented in \cite{systema} are also based on geodesic curves of the Grassmann manifold, although they do not exhibit the property of having a single nonzero entry per row. On the other hand, unstructured designs do not impose any specific restrictions and rely on numerical optimization, typically with the primary goal of improving error probability performance \cite{ub}, \cite{un1}, \cite{un2}, \cite{dp}, in which they generally surpass structured designs. However, they have the disadvantage of yielding non-deterministic solutions, and the absence of structure makes it difficult to reduce the computational complexity of ML detection. To the best of our knowledge, the work in \cite{ub}, based on gradient optimization of the UB criterion, provides the best unstructured Grassmannian constellation designs in terms of reported numerical SER results, followed closely by the earlier work in \cite{dp}, which was based on gradient optimization of the DP (coherence criterion).

Related to this problem, the noncoherent SIMO scenario (which can be considered a particular instance of the MIMO case) has recently garnered considerable attention \cite{cubesplit}, \cite{simo_grass}, \cite{simo_grass2}, \cite{simo_nograss1}, \cite{simo_nograss2}. In this context, unitary (or Grassmannian) signaling has also been shown to be an effective solution in practical scenarios \cite {cubesplit}, \cite{simo_grass}, \cite{simo_grass2}. The inherent simplicity of the SIMO case compared to MIMO enables low-complexity designs that are particularly appealing for uplink communications. Finally, a promising research direction is the design of Grassmannian constellations for noncoherent multi-user MIMO, a field where some recent works have also contributed \cite{joint_geo}, \cite{multi-user}.

\subsection{Contributions}\label{cont}
The main contributions of this work are summarized as follows:

\begin{itemize}
    \item We present a novel family of structured Grassmannian constellations. The proposed solution is built upon the geodesic curves of the Grassmann manifold, thereby exploiting its underlying geometric structure. The design consists of computing geodesics on the Grassmann manifold from a chosen initial point, along with an appropriate set of initial velocities. These geodesics are then mapped to extract constellation points with desirable properties for error performance.

    \item The proposed family is limited in spectral efficiency (with a maximum constellation size of $4M^2$ points, where $M$ is the number of transmit antennas), but offers one significant advantage: all resulting space-time matrices have a single nonzero element per row, i.e., only one active antenna is required per time slot. This has important implications for practical implementation: at the transmitter, only one digital-to-analog converter (DAC) and radio-frequency (RF) chain is required, together with a switch to select the corresponding antenna at each time slot. This not only simplifies the hardware but also reduces power consumption by using a single power amplifier. These differences in complexity, cost, and power consumption are already significant even for $M=2$ transmit antennas. We also present a preliminary analysis of the feasibility of the hardware implementation of the proposed scheme.

    \item Antenna selection in MIMO is a well-known solution in the literature \cite{antenna_selection}, typically offering a tradeoff between reduced hardware complexity and cost on the one hand, and system performance on the other. However, the proposed structured family of Grassmannian constellations achieves comparable symbol error rate (SER) performance to state-of-the-art optimization-based unstructured designs, as validated through numerical simulations.

    \item Additionally, the proposed design reduces the computational complexity of the ML detector by a factor of $M$ and also simplifies bit labeling. To validate the latter point, comparisons between SER and bit error rate (BER) are presented.

    \item Although we again remark that our solution is limited in spectral efficiency, the considered range of constellation points is arguably one of the most interesting for the application of noncoherent MIMO, since larger constellations are subject to a high computational detection cost, especially in unstructured designs aimed at optimizing error performance. In this sense, the presented approach offers relatively low spectral efficiencies ($0.25$–$1$ bps/Hz) but achieves error performance comparable to state-of-the-art solutions for noncoherent MIMO within this range, while enabling low-cost and low-complexity hardware implementation. In fact, the proposed design may find applicability in scenarios where SIMO has previously been considered, offering enhanced error performance by simply adding a switching scheme and multiple antennas at the transmitter, without the need for additional RF chains or power amplifiers. Taking these considerations into account, the proposed family of constellations is not well suited to high-data-rate regimes. However, it is particularly appealing for systems subject to stringent cost constraints (typical of IoT deployments) and/or power constraints (typical, for instance, of unmanned aerial vehicle (UAV) communications), as well as for applications with the strict reliability and latency requirements of URLLC (including safety-critical vehicular communication links), which, in general, impose short-packet, low-data-rate and low spectral efficiency ($\leq 1$ bps/Hz) transmissions.

    \item Finally, we note that, during the review process of this manuscript, \cite{sparse} was published, proposing a method based on Schubert cells for designing sparse Grassmannian constellations capable of satisfying the single-nonzero-element-per-row condition. However, the SER results reported in that work, which are limited to $M=2$ and $M=3$, are outperformed by optimization-based approaches, unlike our proposal. Nevertheless, a remarkable feature of that design is its ability to achieve even higher levels of sparsity.

\end{itemize}

The remainder of this paper is organized as follows. Section \ref{SM} introduces the system model for noncoherent MIMO communications with Grassmannian signaling. Section \ref{SEP} reviews the metrics related to the Grassmann manifold that can predict PEP and SEP and thus guide the design of constellations. Section \ref{cd} describes the proposed constellation design, while Section \ref{res} provides numerical results based on Monte Carlo simulations for SER and BER. Section \ref{HWI} presents a preliminary evaluation of certain hardware implementation aspects. Finally, conclusions are drawn in Section \ref{con}. In addition, Appendix\ref{th} provides the theoretical background on the Grassmann and Stiefel manifolds required for this work. Appendix\ref{wh} shows the construction of the so-called Weyl-Heisenberg basis, which is used in the constellation design. Appendix\ref{pr1}, Appendix\ref{pr2}, and Appendix\ref{pr3} contain the mathematical proofs of the proposed theorems.

\textit{Notation}: Scalars are denoted by italic letters (either lowercase or uppercase, $x$, $X$), while matrices are denoted by bold uppercase letters, $\mathbf{X}$. To avoid confusion with differential geometry notation, we use the notation $\mathbf{X}$ (and refer to it as a matrix) even when $\mathbf{X}$ is a column or row matrix, and reserve the term vector for tangent vectors to a manifold (which are also represented as matrices that fulfill certain conditions). Subspaces and sets are denoted by calligraphic letters, $\mathcal{X}$. The transpose, conjugate transpose, determinant, and trace of a matrix are denoted by $\left(\cdot\right)^T$, $\left(\cdot\right)^H$, $\operatorname{det}\left(\cdot\right)$, and $\operatorname{tr}\left(\cdot\right)$, respectively. The identity matrix of size $M$ is written as $\mathbf{I}_M$, and $\mathbf{0}_M$ denotes a square matrix of size $M$ with all entries equal to zero. The real and imaginary parts of a scalar or matrix are denoted by $\mathfrak{Re}(\cdot)$ and $\mathfrak{Im}(\cdot)$, respectively. The imaginary unit is denoted as $i$. The absolute value of a scalar, or the cardinality of a set, is denoted by $\left|\cdot\right|$. The floor of a scalar $X$, denoted as $\lfloor X \rfloor$, is defined as the unique integer $n \in \mathbb{Z}$ such that 
$n \le X < n + 1$. The expected value is denoted as $\operatorname{E}\{\cdot\}$. The Frobenius norm of a matrix is written as $\|\mathbf{X}\|_F$. A complex Gaussian distribution with mean $\mu$ and variance $\sigma^2$ is denoted by $\mathcal{CN}\left(\mu,\sigma^{2}\right)$. Appendix\ref{th} provides the necessary theoretical background on the Grassmann and Stiefel manifolds, along with their more specific notation.

\section{System Model}\label{SM}
\subsection{System Model}\label{smm}
Consider a MIMO wireless system where a transmitter equipped with $M$ antennas communicates with a receiver with $N$ antennas over a frequency-flat block-fading channel with a coherence of $T$ symbols, $T \geq 2M$. Consequently, it is assumed that the channel matrix $\mathbf{H} \in \mathbb{C}^{M \times N}$ remains constant during each block of $T$ symbols, and changes to an independent realization in the next block. The channel entry in the $i$-th row and $j$-th column is denoted by $h_{ij} \in \mathbb{C}$. All channel entries are i.i.d. and follow a normal distribution with zero mean and unit variance, i.e.,  $h_{ij} \sim  \mathcal{C}\mathcal{N}(0,1)$, according to classical Rayleigh fading. Furthermore, the channel realizations are unknown to both the transmitter and the receiver, a scenario commonly referred to as a noncoherent communication.

Within a block of $T$ symbols, the transmitter sends a matrix $\mathbf{X} \in \mathbb{C}^{T \times M}$, such that $\mathbf{X}^{H}\mathbf{X} = \mathbf{I}_M$. The matrix $\mathbf{X}$ is a Stiefel representative of a point $[\mathbf{X}]$ in the complex Grassmann manifold $\textup{Gr}_{\mathbb{C}}(T,M)$, as described in Appendix\ref{th}. For each time block, the transmit matrix is chosen uniformly from an alphabet $\mathcal{X}=\{ \mathbf{X}_1,\mathbf{X}_2, \ldots,\mathbf{X}_L\}$, so $\left|\mathcal{X}\right| = L$. The alphabet $\mathcal{X}$ is also denoted as the codebook or the constellation for the communication, and its elements, as the constellation points. For a given constellation, the rate of the communication is determined by $R=\log_2(L)/{T}$ (bps/Hz) and each codeword carries $\log_2(L)$ bits of information. 

The signal at the receiver over a time block is denoted as $\mathbf{Y} \in \mathbb{C}^{T \times N}$ and it is given by 

\begin{equation}\label{Y}
\mathbf{Y} = \mathbf{X}\mathbf{H} + \sqrt{\frac{M}{T\rho}}\mathbf{W},
\end{equation} where $\mathbf{W}\in \mathbb{C}^{T \times N}$ represents the additive white Gaussian noise (AWGN), i.e., its entries are i.i.d. with $w_{ij}\sim \mathcal{C}\mathcal{N}(0,1)$ and $\rho$ denotes the signal-to-noise ratio (SNR). Assuming Rayleigh block-fading, the optimal ML detector that minimizes error probability is  \cite{gt2}

\begin{equation}\label{ml}
\hat{\mathbf{X}} = \argmax_{\mathbf{X}\in\mathcal{X}} \operatorname{tr}\left( \mathbf{Y}^H \mathbf{P}_{[\mathbf{X}]}\mathbf{Y}\right),
\end{equation} where $\mathbf{P}_{[\mathbf{X}]} = \mathbf{X}\mathbf{X}^H$ is the projector onto the subspace $\operatorname{span}(\mathbf{X})$ (see Appendix\ref{th}).

\subsection{Grassmmannian Signaling Application}\label{gsa}

The capacity of the noncoherent block-fading channel model under consideration was first studied in \cite{gt1}. Unitary space-time modulation was proposed in \cite{gt2} as a solution to the problem and demonstrated optimal in \cite{gt2_2} in the sense that it minimizes the union bound of error probability under the assumption of equal-energy signals. In this context, note that Grassmannian signaling and unitary space-time modulation refer to the same concept, that is, $\mathbf{X}^H\mathbf{X} = \mathbf{I}_M$, which is equivalent to stating that the columns of $\mathbf{X}$ are orthonormal. In \cite{gt3}, it was shown that, at high SNR, the degrees of freedom (DoF) of the system are $M^*\left(1-\frac{M^*}{T}\right)$, where  $M^*=\min\{M, N,\lfloor T/2\rfloor\}$, and that the use of Grassmannian signaling is optimal for ergodic capacity when $T\geq\min\{M, N\}+N$. Furthermore, for a given $T$, the optimal number of transmit antennas is $\lfloor T/2 \rfloor$ (as additional transmit antennas do not increase capacity), while the number of receive antennas should be no less than $\lfloor T/2 \rfloor$. Along the same lines, the numerical mutual information results presented in \cite{gt4} also concluded that, at high SNR, mutual information is maximized when $M = \min \{N, \lfloor T/2 \rfloor\}$.

In the case $T < M + N$, Grassmannian signaling is no longer optimal \cite{gt5}, since it does not achieve capacity at high SNR. The corresponding capacity-achieving input signal, introduced in \cite{gt5}, is referred to as beta-variate space-time modulation and does not satisfy the equal-energy condition, as different transmitted matrices may have different energy levels. Nevertheless, for beta-variate space-time modulation, numerical results also indicate that the optimal number of transmit antennas is $M = \lfloor T/2 \rfloor$, while the rate gain becomes significant compared to Grassmannian signaling only when $N \gg T$.

The structured family of Grassmannian constellations presented in this paper is designed for any number of transmit antennas $M = 1,2,3,4,\ldots$, with $T = 2M$, and for any number of receive antennas, $N$. Based on previous results, the most suitable application scenario is when $T = M + N = 2M$, i.e., $N = M$ and $M = T/2$, in which case Grassmannian signaling is the optimal choice for noncoherent communication in the Rayleigh block-fading channel. Nonetheless, in Section \ref{res} SER and BER results are presented for different numbers of receive antennas for a given $M$, while maintaining $T = 2M$, in order to illustrate error performance differences consistent with the mutual information results reported in the literature. 

Note that the condition $T=2M$ can always be imposed by the system. For instance, if in a practical implementation the channel remains quasi-constant during an odd number of symbols, e.g., $T=7$, the system can transmit using a coherence block of $T=6$ with $M=3$ transmit antennas. Another option would be to use the full coherence block with $T = 7$, setting $M = \lfloor T/2 \rfloor = 3$ and $N = 4$, so that $M + N = T$, in which case Grassmannian signaling remains optimal. However, the increase in the DoF of the system is rather limited, according to the previously presented expression $M^*\left(1-\frac{M^*}{T}\right)$, with $M^*=\min\{M, N,\lfloor T/2\rfloor\}$ \cite{gt3}. Furthermore, the applicability of the scheme (even $T$ and $M$, with $M$ greater than $1$ but different from $N$) would be limited, and increasing $N$ beyond $\lfloor T/2 \rfloor$, which does not increase the DoF but provides array diversity, can always be applied whether $T$ is even or not.

\section{Metrics and Error Probability} \label{SEP}
The relationship between various metrics related to the Grassmann manifold and the error probability of a given constellation in the block-fading channel described in Section \ref{SM} has been extensively studied in the literature \cite{do1}, \cite{do2}, \cite{do3}, \cite{ubp}, \cite{ub}, \cite{dp}. Most of these metrics are related to the principal angles between two points on the Grassmann manifold, $[\mathbf{X}_1]$ and $[\mathbf{X}_2]$, defined as 

\begin{equation}\label{pans3}
        \theta_m = \arccos{(\sigma_m)} \in \left[0,\frac{\pi}{2}\right], \quad m = 1,\ldots, M,
\end{equation} where $\sigma_m\in [0,1]$ is the $m$-th largest singular value of $\mathbf{X}_1^H\mathbf{X}_2$, and the result is independent of the choice of the Stiefel representatives $\mathbf{X}_1$ and $\mathbf{X}_2$. The principal angles determine the geodesic distance between two points, $\textrm{d}_{\textrm{g}}$, given by

\begin{equation}
    \label{eq:geod_dists3}
        \textrm{d}_\textrm{g}\left([\mathbf{X}_1],[\mathbf{X}_2]\right) = \left(\sum_{m=1}^M \theta_m^2 \right)^{1/2},
\end{equation} which corresponds to the natural notion of Riemannian distance, as explained in Appendix\ref{th}, where the concepts of principal angles and geodesic distance are introduced in a more detailed and rigorous manner, even at the risk of some minor redundancy. Note that, strictly speaking, points in the Grassmannian are subspaces and should be denoted as $[\mathbf{X}]$. However, from this point onward, when computing distances or metrics between two constellation points, we omit the brackets for simplicity and simply write $\mathbf{X}$.

One of the most widely used metrics for designing Grassmannian constellations for noncoherent communications is the chordal distance \cite{do1}, \cite{grassma}, which we denote by $\textrm{d}_\textrm{c}$ and is defined as

\begin{equation}\label{dc}
\begin{aligned}
\textrm{d}_\textrm{c}\left( \mathbf{X}_i, \mathbf{X}_j \right) &= \frac{1}{\sqrt{2}}\| \mathbf{P}_{[\mathbf{X}_i]}-\mathbf{P}_{[\mathbf{X}_j]} \|_F \\
&= \left(\sum_{m=1}^{M}\sin^2\theta_m\right)^{1/2} ,
\end{aligned}
\end{equation} where, based on Eq. \eqref{pans3}, it is clear that the maximum value for the chordal distance is $\sqrt{M}$. Compared to the geodesic distance, this metric has the numerical advantage of being differentiable at any point on the manifold.

However, the study presented in \cite{do3} revealed that the chordal distance between two constellation points is related to their PEP only at low SNR, while at high SNR, the so-called DP is a better criterion for minimizing the PEP between two codewords. The DP between two constellation points is defined as \cite{do2}, \cite{do3}

\begin{equation}\label{dp}
\begin{aligned}
\textrm{DP}(\mathbf{X}_i, \mathbf{X}_j) &= \textrm{det} \left(\mathbf{I}_M - \mathbf{X}_i^H\mathbf{X}_j\mathbf{X}_j^H\mathbf{X}_i  \right)\\
&= \prod_{m=1}^M \sin^2\theta_m.
\end{aligned}
\end{equation}

In general, when designing a constellation based on either the chordal distance or DP, the goal is to maximize the minimum value across all possible pairs. In this context, it is useful to define the chordal distance of a constellation $\mathcal{X}$ as

\begin{equation}\label{dcc}
\textrm{d}_\textrm{c}\left( \mathcal{X} \right) = \min_{i\neq j} \; \textrm{d}_\textrm{c}\left( \mathbf{X}_i, \mathbf{X}_j \right),
\end{equation} and the DP of a constellation as

\begin{equation}\label{dpc}
\textrm{DP}\left( \mathcal{X} \right) = \min_{i\neq j} \; \textrm{DP}\left( \mathbf{X}_i, \mathbf{X}_j \right),
\end{equation} where the objective of minimizing $\textrm{DP}\left(\mathcal{X}\right)$ is also known as the coherence criterion \cite{dp}.

Based on the previous explanation, note that the chordal and DP approaches focus on improving the worst-case PEP. This strategy is standard for Gaussian channels affected only by AWGN, where the SEP is typically dominated by the largest PEP. However, under Rayleigh fading, particularly at medium or low SNR, considering codeword pairs beyond the worst case may significantly improve the accuracy of SEP approximations. In \cite{ubp} and \cite{ub}, the asymptotic UB is presented as a metric that accounts for all codeword pairs when evaluating a constellation. In \cite{ub}, the UB at high SNR is defined, up to a constant, as

\begin{equation}\label{ubc}
\textrm{UB}(\mathcal{X}) = \sum_{i<j} \textrm{det} \left(\mathbf{I}_M - \mathbf{X}_i^H\mathbf{X}_j\mathbf{X}_j^H\mathbf{X}_i  \right) ^{-N}, 
\end{equation} where, remarkably, the number of receive antennas, $N$, appears explicitly in the expression. Minimizing the UB of the constellation is desirable, as it has been proven in \cite{ub} and \cite{aerr} to have the same high-SNR exponent as the error probability. Nonetheless, as indicated in \cite{sparse}, when $\left(\mathbf{I}_M - \mathbf{X}_i^H \mathbf{X}_j \mathbf{X}_j^H \mathbf{X}_i\right)$ is not full rank, the determinant vanishes, which makes certain error-probability expressions based on the UB inadequate for those cases.

In \cite{ub}, a gradient descent optimization is presented to minimize the UB for arbitrary values of $(T,M,N,L)$. To the best of our knowledge, the work in \cite{ub} provides the best Grassmannian constellation designs in terms of reported numerical SER results. Table \ref{table1} summarizes the considered metrics.

\renewcommand{\arraystretch}{1.5} 
\renewcommand\cellgape{\Gape[4pt]} 
\begin{table*}[t]
\caption{Metrics and relationship to PEP.}
\centering
\begin{tabular}{|c|c|c|c|c|c|c|c|c|}
\hline
\textbf{Metric} & \textbf{Definition} & \textbf{Range} & \textbf{Goal} & \textbf{PEP relationship} \\ \hline \hline \makecell{Geodesic distance}& \makecell{$\textrm{d}_\textrm{g}\left(\mathbf{X}_i,\mathbf{X}_j\right) = \left(\sum_{m=1}^M \theta_m^2 \right)^{1/2}$ \\ [6pt] $\textrm{d}_\textrm{g}\left( \mathcal{X} \right) = \min_{i\neq j} \; \textrm{d}_\textrm{g}\left( \mathbf{X}_i, \mathbf{X}_j \right)$} & $\left[0,\sqrt{M}\frac{\pi}{2}\right]$ & - & -  \\ \hline

\hline \makecell{Chordal distance}& \makecell{ $\textrm{d}_\textrm{c}\left(\mathbf{X}_i,\mathbf{X}_j\right) = \left(\sum_{m=1}^M \sin^2\theta_m \right)^{1/2}$ \\ [6pt] $\textrm{d}_\textrm{c}\left( \mathcal{X} \right) = \min_{i\neq j} \; \textrm{d}_\textrm{c}\left( \mathbf{X}_i, \mathbf{X}_j \right)$ } & $\left[0,\sqrt{M}\right]$ & Max. min.  & Worst pair, low SNR\\ \hline

\hline \makecell{Diversity product (DP)}& \makecell{ $\textrm{DP}\left(\mathbf{X}_i,\mathbf{X}_j\right) = \prod_{m=1}^M \sin^2\theta_m$ \\ [6pt] $\textrm{DP}\left( \mathcal{X} \right) = \min_{i\neq j} \; \textrm{DP}\left( \mathbf{X}_i, \mathbf{X}_j \right)$ } & $\left[0,1\right]$ & Max. min.  & Worst pair, high SNR \\ \hline

Union bound (UB) & \makecell{$\textrm{UB}(\mathcal{X}) = \sum_{i<j} \textrm{det} \left(\mathbf{I}_M - \mathbf{X}_i^H\mathbf{X}_j\mathbf{X}_j^H\mathbf{X}_i  \right) ^{-N}$} &  $[1,\infty)$ & Minimize  & All pairs, high SNR \\ \hline

\end{tabular}
\label{table1}
\end{table*}

\section{Constellation Design}\label{cd}

In this section, we present the algorithm for designing Grassmannian constellations using geodesic mapping. The input parameters are the number of transmit antennas, $M$ (with $T = 2M$, as described in Section \ref{gsa}), and the number of constellation points, $L$, where $L$ is a power of two that satisfies $2 \leq L \leq 4M^2$. The goal is to obtain a constellation $\mathcal{X}$ with two properties: first, all constellation points must satisfy that each row of the transmitted space-time matrix contains a single nonzero entry, i.e., only one antenna is active per time slot; second, error performance at least comparable to previous constellation designs, as validated in Section \ref{comp}. Additionally, the algorithm inherently simplifies the problem of bit labeling, as explained in Section \ref{bl}, and the second property reduces the computational complexity of the ML detector, as detailed in Section \ref{rcc}.

As stated in Section \ref{cont}, the spectral efficiency is limited by the maximum number of points, $4M^2$, as shown in Table \ref{table2}, since $R = \log_2(L)/T$ (see Section \ref{SM}). However, it should be noted that higher spectral efficiency values with arbitrary Grassmannian constellations are constrained by the computational cost at the receiver when ML detection is applied. As explained in Section \ref{rcc}, the computational cost of the ML detector for arbitrary Grassmannian constellations is $\mathcal{O}(LMTN)$, which quickly becomes dominated by $L$ as spectral efficiency increases. For example, consider the case $T=8$ and $M=4$ with a spectral efficiency of $1.5$ bps/Hz. This requires $L=4096$ points, exceeding the maximum considered even for unstructured designs (e.g., up to $L=2048$ points in \cite{ub}). This further reinforces the fact that rates in the range of $0.25$–$1$ bps/Hz are particularly relevant for the application of Grassmannian signaling in MIMO noncoherent schemes. Although designs with significantly higher spectral efficiency have been proposed, particularly for the SIMO case \cite{simo_grass}, they achieve this at the expense of reduced error performance, as ML detection is not applied.

\begin{table}[t]
\caption{Maximum number of points and spectral efficiency.}
\centering
\begin{tabular}{|c|c|c|}
\hline
$(T,M)$ & \textbf{\makecell{Maximum \\ number of points}} &  \textbf{\makecell{Maximum \\ spectral efficiency}} \\ \hline \hline

$(2,1)$ & $4$ & $1$ bps/Hz \\ \hline

$(4,2)$  & $16$ & $1$ bps/Hz\\ \hline

$(6,3)$  & $36$ $(32)$ & $5/6$ bps/Hz \\ \hline

$(8,4)$  & $64$ & $0.75$ bps/Hz \\ \hline

$(16,8)$  & $256$ & $0.5$ bps/Hz \\ \hline

\end{tabular}
\label{table2}
\end{table}

\subsection{Geodesics and Diametral Sets}\label{vgg}
We now present the mathematical results and the specific choices on which our design is based. We refer the reader to Appendix\ref{th} for the definitions of the mathematical objects and notions used here.

\begin{theorem}
\label{thm:geodesic}
Let $[\mathbf{U}] \in \textup{Gr}_{\mathbb{C}}(2M,M)$ be a point and $\mathbf{\Delta} \in T_{[\mathbf{U}]}\textup{Gr}_{\mathbb{C}}(2M,M)$ a vector of the form
\begin{equation}
\label{eq:tilde_matrix}
    \mathbf{U} =
    \begin{pmatrix}
    \tilde{\mathbf{U}} \\
    \mathbf{0}_M
    \end{pmatrix}, \quad
    \mathbf{\Delta} =
    \begin{pmatrix}
    \mathbf{0}_M \\
    \tilde{\mathbf{\Delta}}
    \end{pmatrix},
\end{equation}
where $\tilde{\mathbf{U}} \in \textup{St}_{\mathbb{C}}(M,M)$ and $\tilde{\mathbf\Delta} \in \mathbb{C}^{M\times M}$. If $\sqrt{M}\tilde{\mathbf{\Delta}} \in U(M)$, then $\mathbf{\Delta}$ is a diametral vector and the geodesic $\gamma_{\mathbf{\Delta}}$ admits the following expression:
\begin{equation}
\label{eq:geodesic}
    \gamma_{\mathbf{\Delta}}(t) =
    \begin{bmatrix}
    \begin{pmatrix}
    \cos{\left(\frac{t}{\sqrt{M}}\right)}\tilde{\mathbf{U}} \\
    \sqrt{M}\sin{\left(\frac{t}{\sqrt{M}}\right)}\tilde{\mathbf{\Delta}}
    \end{pmatrix}
    \end{bmatrix}, \quad \forall t \in \mathbb{R}.
\end{equation}
\end{theorem}
\begin{proof}
See Appendix\ref{pr1}.
\end{proof}
The key implication of this result is that, if we choose $\tilde{\mathbf{U}}$ and $\tilde{\mathbf{\Delta}}$ so that each row contains only one nonzero entry, then each geodesic point admits a Stiefel representative that inherits this property. Our proposal is therefore to obtain the constellation points from Eq. \eqref{eq:geodesic} by carefully selecting one point $\mathbf{U}$ and $4M^2$ vectors $\mathbf{\Delta}$ in the form of Eq. \eqref{eq:tilde_matrix}.

To select suitable vectors, we first remark that, for any $\mathbf{U}$ in the form of Eq. \eqref{eq:tilde_matrix}, it is always possible to construct a basis $\{\mathbf{\Delta}_k\}_{k=1}^{2M^2}$ of $T_{[\mathbf{U}]}\textup{Gr}_{\mathbb{C}}(2M,M)$ satisfying the following properties:
\begin{itemize}
    \item Each $\mathbf{\Delta}_k$ has the form of Eq. \eqref{eq:tilde_matrix} and $\sqrt{M}\tilde{\mathbf{\Delta}}_k \in U(M)$, i.e., $\mathbf{\Delta}^H\mathbf{\Delta} = \tilde{\mathbf{\Delta}}^H\tilde{\mathbf{\Delta}} = \frac{1}{M}\mathbf{I}_M = \tilde{\mathbf{\Delta}}\tilde{\mathbf{\Delta}}^H$. In particular, this implies that every $\mathbf{\Delta}_k$ is $g$-unit (with respect to the Riemannian metric $g$ in Eq. \eqref{eq:riem_metric}).
    \item $\{\mathbf{\Delta}_k\}_{k=1}^{2M^2}$ is $g$-orthonormal, i.e., $g_{[\mathbf{U}]}(\mathbf{\Delta}_k,\mathbf{\Delta}_l)=\delta_{kl}$, where $\delta_{kl}$ is the Kronecker delta.
    \item Each $\tilde{\mathbf{\Delta}}_k$ has only one nonzero entry in each row.
\end{itemize}
The first condition ensures that every vector of the basis satisfies the hypotheses of Theorem \ref{thm:geodesic}. The second condition, together with the fact that each $\mathbf{\Delta}_k$ is diametral (by Theorem \ref{thm:geodesic}), allows us to fully exploit the distribution of the geodesics over the manifold. The last condition, if also satisfied by $\tilde{\mathbf{U}}$, guarantees that the constellation points obtained from Eq. \eqref{eq:geodesic} have exactly one nonzero entry per row. One example of such a basis can be constructed by taking as $\{\tilde{\mathbf{\Delta}}_k\}_{k=1}^{2M^2}$ the ($g$-orthonormal) Weyl-Heisenberg basis, given by Eq. \eqref{eq:wh_basis} (see Appendix\ref{wh}). Finally, observe that given $\{\mathbf{\Delta}_k\}_{k=1}^{2M^2}$, their opposites $\{-\mathbf{\Delta}_k\}_{k=1}^{2M^2}$ also form a basis of $T_{[\mathbf{U}]}\textup{Gr}_{\mathbb{C}}(2M,M)$ with the same properties.

Taking all this into account, we make the following explicit choices for $\mathbf{U}$ and $\mathbf{\Delta}$ of the form of Eq. \eqref{eq:tilde_matrix}: since there are no privileged points in the Grassmannian, we set $\tilde{\mathbf{U}} = \mathbf{I}_M$ for simplicity, and we select the $4M^2$ vectors $\{\pm\mathbf{\Delta}_k\}_{k=1}^{2M^2}$, where $\{\tilde{\mathbf{\Delta}}_k\}_{k=1}^{2M^2}$ is the Weyl-Heisenberg basis. In the constellation design, each vector provides a single constellation point by choosing a suitable $t \in \mathbb{R}$ in Eq. \eqref{eq:geodesic}, yielding a maximum of $4M^2$ constellation points.

For the DP and UB metrics (Section \ref{SEP}), we are interested in constellation points whose pairwise principal angles are all nonzero. At the same time, we must also ensure good performance with respect to the geodesic and chordal distances by distributing the constellation points as far apart as possible over the manifold. To identify such points, we make use of the following results.

\begin{theorem}
\label{thm:principal_angles}
    Within the hypotheses of Theorem \ref{thm:geodesic}, let $\gamma_{\mathbf{\Delta}_1}(t_1)$ and $\gamma_{\mathbf{\Delta}_2}(t_2)$ be two geodesic points satisfying Eq. \eqref{eq:geodesic}, with $t_1, t_2 \in \left( 0,\sqrt{M}\frac{\pi}{2} \right)$. Then:
    \begin{itemize}
        \item If $t_1 \not= t_2$, all the principal angles between both geodesic points are nonzero.
        \item If $t_1 = t_2$, the number of principal angles equal to zero between both geodesic points coincides with the number of eigenvalues of $M\mathbf{\Delta}_1^H \mathbf{\Delta}_2$ equal to $1$.
    \end{itemize}
\end{theorem}
\begin{proof}
    See Appendix\ref{pr2}.
\end{proof}

\begin{corollary}
\label{cor:opposite}
    Let $\mathbf{\Delta} \in T_{[\mathbf{U}]}\textup{Gr}_{\mathbb{C}}(2M,M)$ be a vector satisfying the hypotheses of Theorem \ref{thm:geodesic}. Then, all the principal angles between $\gamma_{\mathbf{\Delta}}\left(\sqrt{M}\frac{\pi}{4}\right)$ and $\gamma_{-\mathbf{\Delta}}\left(\sqrt{M}\frac{\pi}{4}\right)$ are equal to $\frac{\pi}{2}$.
\end{corollary}
\begin{proof}
    See Appendix\ref{pr3}.
\end{proof}

This means that geodesic points associated with opposite vectors, when taken halfway along the geodesic, exhibit exactly the properties we are looking for: all their principal angles are maximal (in particular, nonzero), which in turn implies that they are separated by the maximum possible geodesic and chordal distances in the Grassmann manifold. Therefore, in our design, we are interested in obtaining constellation points from vector pairs $\pm\mathbf{\Delta}$ whenever possible, while ensuring at the same time that constellation points coming from different pairs do not have any principal angle equal to zero. With this in mind, we say that $\mathcal{G} \subset \{\pm\mathbf{\Delta}_k\}_{k=1}^{2M^2}$ is a {\em diametral set} if it satisfies the following conditions:
\begin{itemize}
    \item For any two vectors $\mathbf{\Delta}_1, \mathbf{\Delta}_2 \in \mathcal{G}$, there is no eigenvalue of $M\mathbf{\Delta}_1^H\mathbf{\Delta}_2$ equal to $1$.
    \item If $\mathbf{\Delta} \in \mathcal{G}$, then $-\mathbf{\Delta} \in \mathcal{G}$.
\end{itemize}
By Theorem \ref{thm:principal_angles}, the first condition means that all the principal angles between $\gamma_{\mathbf{\Delta}_1}(t)$ and $\gamma_{\mathbf{\Delta}_1}(t)$ are nonzero for all $t \in \left( 0,\sqrt{M}\frac{\pi}{2} \right)$. The second condition guarantees that, within the same diametral set $\mathcal{G}$, we can exploit the properties of geodesic points associated with opposite vectors $\pm\mathbf{\Delta} \in \mathcal{G}$. Note that both conditions are consistent, thanks to the fact that all the eigenvalues of $-M\mathbf{\Delta}_k^H\mathbf{\Delta}_k=-\mathbf{I}_m$ are equal to $-1$.

Given $M$, we denote by $D$ the cardinality of the largest diametral set that can be constructed from $\{\pm\mathbf{\Delta}_k\}_{k=1}^{2M^2}$. Observe that $D\geq 4$, since $\{\pm\mathbf{\Delta},\pm i\mathbf{\Delta}\}$ is always a diametral set for any $\mathbf{\Delta} \in \{\pm\mathbf{\Delta}_k\}_{k=1}^{2M^2}$. However, in order to achieve the maximal size $D$, it may be optimal to separate $\pm\mathbf{\Delta}$ and $\pm i\mathbf{\Delta}$ into different diametral sets.

\subsection{Algorithm Description} \label{ad}
Given the background developed in Section \ref{vgg}, we now present our proposed constellation design algorithm, Algorithm \ref{algc}. Note that the maximum size of a diametral set, $D$, is deterministic and determined by $M$ (which defines $\textup{Gr}_{\mathbb{C}}(2M,M)$) and the chosen vector basis. As an illustrative example of the algorithm, consider the case $M=2$ (thus $T=4$), which yields $D=8$. Table \ref{table3} presents the obtained results for all possible values of $L$ in terms of $\textrm{d}_\textrm{g}\left( \mathcal{X} \right)$, $\textrm{d}_\textrm{c}\left( \mathcal{X} \right)$, and $\textrm{DP}\left( \mathcal{X} \right)$. Importantly, in all tested scenarios with $M>1$, $D \leq 2M^2$. For $M=1$, $D=4$ and only Case (i) and Case (ii) in Algorithm \ref{algc} apply. Furthermore, in all cases, $D$ does not depend on the initial pair chosen to construct the first diametral set. These aspects are particularly relevant since Algorithm \ref{algc} requires at most two diametral sets such that $\mathcal{G}_1 \cap \mathcal{G}_2 = \emptyset$. Altogether, this implies that they can be obtained through a basic combinatorial search that begins by choosing any pair $\pm\mathbf{\Delta}$ and finding the maximum value of $D$. Finally, note that the different cases in Algorithm \ref{algc} reflect an effort to present a systematic constellation construction for any value of $M$ and $L\leq 4M^2$, while consistently achieving the best error performance we have identified by exploiting the geometric structure of the Grassmann manifold.

\begin{algorithm}
\caption{Constellation Design}
\label{algc}
\begin{algorithmic}[1]
    \REQUIRE $M$ (which defines $T$ and $D$), and $L$, $2\leq L \leq 4M^2$ 
    \ENSURE Constellation $\mathcal{X}$, $|\mathcal{X}|=L$

    \IF{$L=2$}
        \STATE \COMMENT{Case (i)}
        \STATE Select one $\mathbf{\Delta}_k$ and take the two points $\gamma_{\mathbf{\Delta}_k}(0)$ and $\gamma_{\mathbf{\Delta}_k}\left(\sqrt{M}\frac{\pi}{2}\right)$
    \ELSIF{$L=4$}
        \STATE \COMMENT{Case (ii)}
        \STATE Select two vector pairs in the same diametral set $\mathcal{G}$, $\pm\mathbf{\Delta}_1$ and $\pm\mathbf{\Delta}_2$, and adjust geodesic mapping taking the four points $\gamma_{\pm\mathbf{\Delta}_1}\left(\sqrt{M}\frac{\pi}{4}+x\right)$, $\gamma_{\pm\mathbf{\Delta}_2}\left(\sqrt{M}\frac{\pi}{4}-x\right)$
    \ELSIF{$4<L\leq D$}
        \STATE \COMMENT{Case (iii)}
        \STATE Select any diametral set $\mathcal{G}$ of size $L$, and take the points $\gamma_{\pm\mathbf{\Delta}_k}\left(\sqrt{M}\frac{\pi}{4}\right)$
    \ELSIF{$D<L\leq 2D$}
        \STATE \COMMENT{Case (iv)}
        \STATE Select any two diametral sets $\mathcal{G}_1$ and $\mathcal{G}_2$ of size $L/2$ without any repeated vectors, $\mathcal{G}_1 \cap \mathcal{G}_2=\emptyset$, and adjust geodesic mapping taking the points $\gamma_{\mathbf{\Delta}_1}\left(\sqrt{M}\frac{\pi}{4}+x\right)$  for all $\mathbf{\Delta}_1\in \mathcal{G}_1$, and $\gamma_{\mathbf{\Delta}_2}\left(\sqrt{M}\frac{\pi}{4}-x\right)$ for all $\mathbf{\Delta}_2\in \mathcal{G}_2$
    \ELSE  
        \STATE \COMMENT{Case (v)} 
        \STATE Take the points $\gamma_{\pm\mathbf{\Delta}_k}\left(\sqrt{M}\frac{\pi}{4}\right)$ for any $L/2$ vector pairs $\pm\mathbf{\Delta}_k$
    \ENDIF
\end{algorithmic}
\end{algorithm}

\begin{table}[t]
\caption{Algorithm performance for $T=4, M=2$}
\centering
\begin{tabular}{|c|c|c|c|}
\hline
$L$ & $\textrm{d}_\textrm{g}\left( \mathcal{X} \right)$ & $\textrm{d}_\textrm{c}\left( \mathcal{X} \right)$ & $\textrm{DP}\left( \mathcal{X} \right)$ \\ \hline \hline

$2$ & $\sqrt{2}\frac{\pi}{2}$ & $\sqrt{2}$ & $1$  \\ \hline

$4$ & $1.3508$ & $1.1546$ & $0.4442$  \\ \hline

$8$ & $1.1107$      & $1$ & $0.25$ \\ \hline

$16$ & $0.9888$      & $0.9102$ & $0.1715$ \\ \hline

\end{tabular}
\label{table3}
\end{table}

For Case (i) in Algorithm \ref{algc}, the solution is optimal with respect to all metrics in Table \ref{table1}, since all principal angles between the two points are equal to $\frac{\pi}{2}$. For Case (iii), the diametral set ensures that no zero principal angles exist between points in the constellation, and, additionally, mapping all geodesics at their midpoint provides good performance in terms of chordal distance. In Case (ii) and Case (iv), the geodesic mapping is adjusted through a single parameter, $x$. For Case (ii), compared to Case (iii), having only two vector pairs allows an adjustment that can be used to further enhance performance. In Case (iv), this adjustment becomes necessary to ensure a nonzero DP, since two diametral sets are involved. Both cases are illustrated in Fig. \ref{fs} for $T=4$, $M=2$, where the dotted line indicates the value of $x$ that minimizes the UB when $N=2$. Remarkably, in Fig. \ref{fsa} (Case (ii)), the optimal mapping point is the same across all metrics. In contrast, in Fig. \ref{fsb}, choosing the value of $x$ that maximizes DP slightly reduces the chordal and geodesic distances, though this value remains close to the one that minimizes UB. In general, for Case (ii) and Case (iv) with any $T$ and $M$, we select the value of $x$ that maximizes DP. In all observed scenarios, this $x$ is close to the value that minimizes UB (although this value varies slightly with $N$). As a result, choosing DP instead of UB has only a minor impact on the resulting constellation points and, consequently, a negligible effect on error performance (see Section \ref{ubdp}). Furthermore, \cite{sparse} showed that the UB expression in Eq. (\ref{ubc}) may not always be fully adequate for sparse matrices, which further justifies our default choice of the DP criterion.

Finally, for Case (v), numerical results obtained by adjusting geodesic mapping for more than two diametral sets achieve poor performance compared to gradient-based optimization approaches in terms of either UB \cite{ub} or DP \cite{dp}. Consequently, we forgo DP or UB optimization and instead adopt an approach focused on chordal distance. Remarkably, when mapping all $L=4M^2$ geodesics at $\sqrt{M}\frac{\pi}{4}$, the minimum chordal distance between points is $\sqrt{M}/\sqrt{2}$, i.e., the maximum possible chordal distance in the Grassmann manifold scaled by $1/\sqrt{2}$. Although not mathematically proven here, this result has been validated numerically for $M=1,2,\ldots,8$ and demonstrates excellent chordal distance performance for $L=4M^2$ points, leading to strong error performance at low SNR (see Section \ref{comp}).

Furthermore, this result is particularly noteworthy considering the Rankin simplex bound on chordal distance. Assuming $T=2M$, the Grassmann manifold $\textup{Gr}_{\mathbb{C}}(T,M)$ can be isometrically embedded into the sphere $ \mathbb{S}^{T^2-2} \subset \mathbb{R}^{T^2-1}$ of radius $\sqrt{M}/2$ \cite{sphere_grass}. Under these conditions, the Rankin simplex bound for spherical codes applies for $L \leq T^2=4M^2$, and it is given by $\textrm{d}_{\textrm{c}} \leq \sqrt{\frac{M}{2}\frac{L}{L-1}}$ \cite{packgrass}. Fig. \ref{bounds} compares the performance of Case (v) in terms of minimum pairwise chordal distance for $L = 4M^2$ points against this bound, showing their proximity. Note that $M \geq 3$ is considered, since Case (v) is not applied in Algorithm \ref{algc} for $M \leq 2$.

 \begin{figure}
\centering
 \subfloat[\normalsize $L=4$, Case (ii) \label{fsa}]{%
       \includegraphics[width=0.8\linewidth]{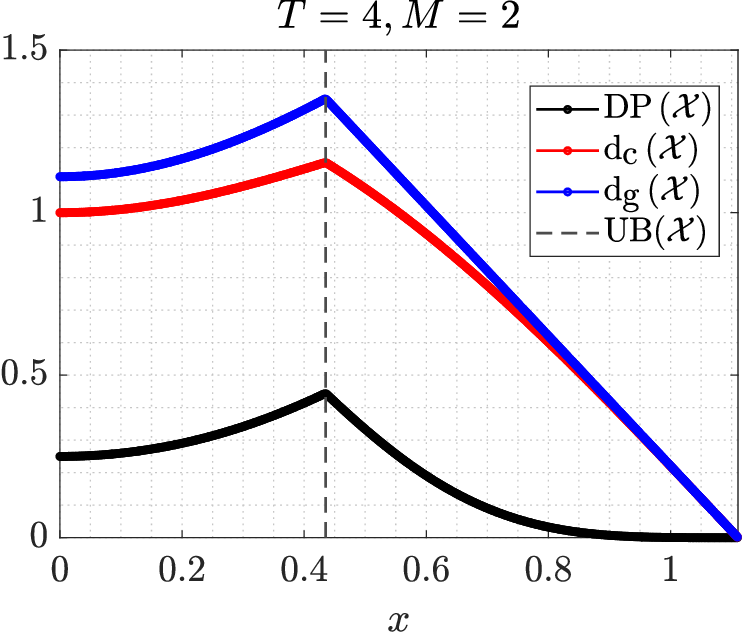}}
  \hfill
  \subfloat[\normalsize $L=16$, Case (iv)\label{fsb}]{%
        \includegraphics[width=0.8\linewidth]{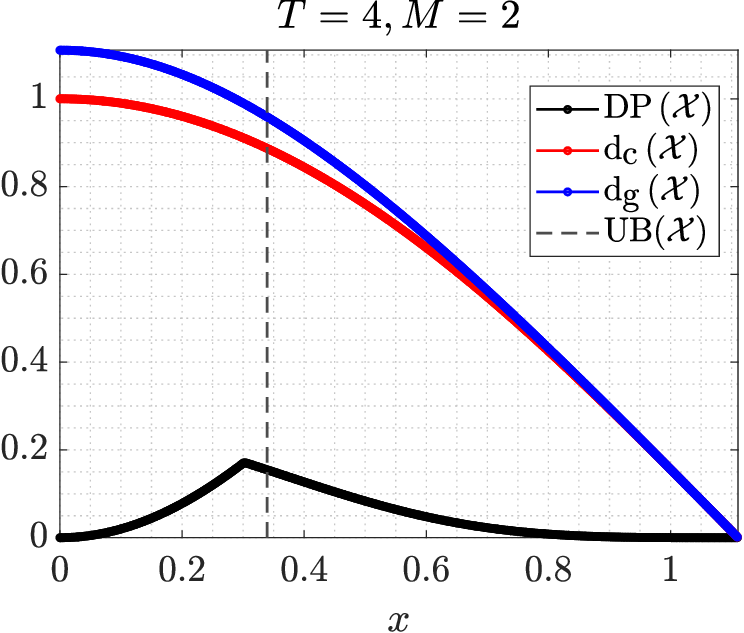}}
  \hfill
  \caption{Constellation metrics versus geodesic mapping parameter, $x$, for Case (ii) and Case (iv) of Algorithm \ref{algc} with $T=4$ and $M=2$. For the UB, the dotted line indicates the value of $x$ that achieves its minimum when $N=2$. A total of $5000$ values of $x$ are shown.}
  \label{fs}
\end{figure}

\begin{figure}
\centering
\includegraphics[width=0.8\linewidth]{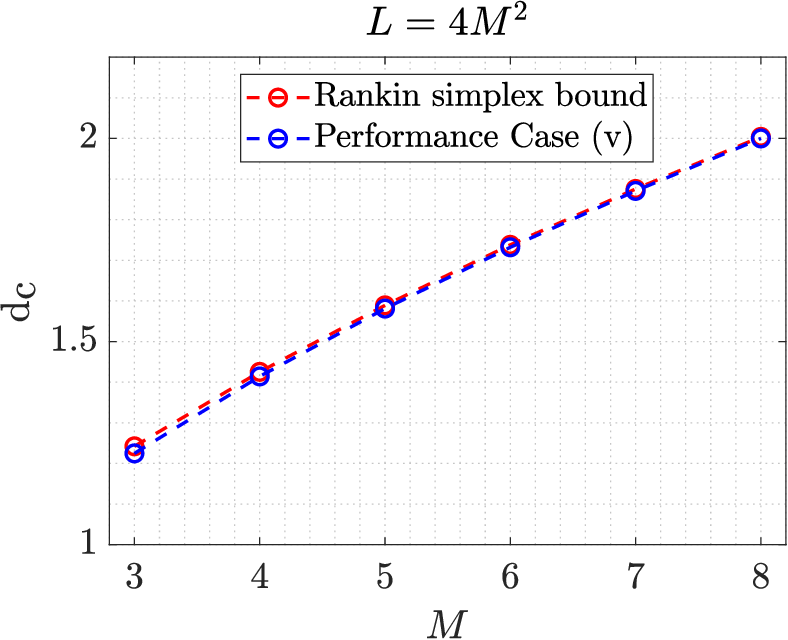}
\caption{Performance of Case (v) in Algorithm \ref{algc} in terms of minimum pairwise chordal distance compared against the Rankin simplex bound for spherical codes.}
\label{bounds}
\end{figure}

\subsection{Bit Labeling}\label{bl}

Once a constellation with the target number of points $L$ is obtained according to Algorithm \ref{algc}, the proposed bit labeling assigns pairs of codewords with the largest Hamming distance to points mapped on geodesics departing in opposite directions ($\pm\mathbf{\Delta}_k$). This assignment is always feasible, since vectors are selected in pairs for any even $L$, and the specific pairing between codewords and point pairs can be chosen arbitrarily. This guarantees that all principal angles between points with maximum Hamming distance are maximal within the constellation (and always nonzero, Section \ref{vgg}) and that such pairs achieve the lowest PEP.

Fig. \ref{geofig} illustrates the proposed point selection and bit labeling for $T = 2$, $M = 1$, and $L = 4$, which corresponds to Case (ii) in Algorithm \ref{algc}, since for $M = 1$, $D = 4$ (Section \ref{ad}). This particular case is appealing for a graphical representation because the complex Grassmannian $\textup{Gr}_{\mathbb{C}}(2,1)$ is isometric (i.e., geometrically equivalent) to the sphere $ \mathbb{S}^2 \subset \mathbb{R}^3 $ of radius $1/2$. In Fig. \ref{geofig}, the point $[\mathbf{U}] \in \textup{Gr}_{\mathbb{C}}(2,1)$, defined as in Eq. \eqref{eq:tilde_matrix} with $\tilde{\mathbf{U}}=(1)$, is identified with the north pole, and the vectors $\mathbf{\Delta}_1, \mathbf{\Delta}_2 \in T_{[\mathbf{U}]}\textup{Gr}_{\mathbb{C}}(2,1)$, defined as in Eq. \eqref{eq:tilde_matrix} with $\tilde{\mathbf{\Delta}}_1=(1)$ and $\tilde{\mathbf{\Delta}}_2=(i)$, are tangent vectors to the sphere at the north pole. In this case, with $L=D=4$, there exists a unique diametral set $\{\pm \mathbf{\Delta}_1,\pm \mathbf{\Delta}_2\}$. The geodesics $\gamma_{\mathbf{\pm\Delta}_1}$ are mapped at $t=1.0931$, while the geodesics $\gamma_{\mathbf{\pm\Delta}_2}$ are mapped at $t=0.4777$.

\begin{figure}
\centering
\includegraphics[width=0.8\linewidth]{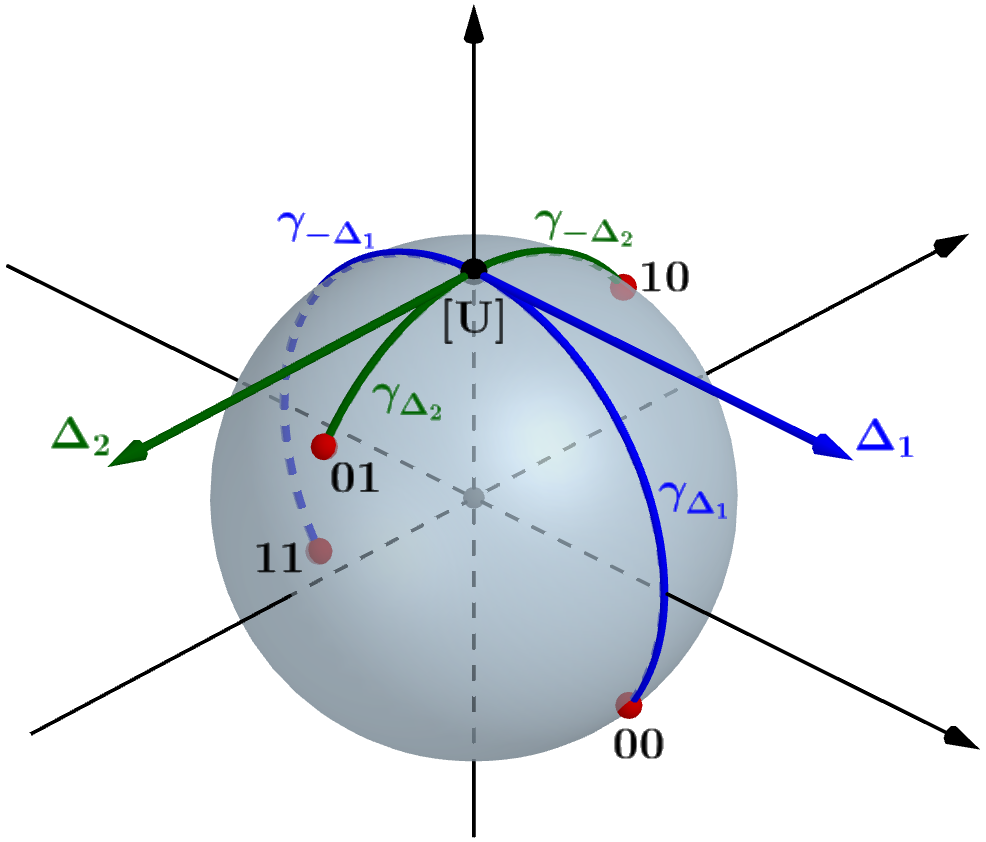}
\caption{Representation of point selection and bit labeling for $T=2$, $M=1$, and $L=4$ (Case (ii) in Algorithm \ref{algc}).}
\label{geofig}
\end{figure}

\subsection{Receiver Computational Complexity}\label{rcc}
Assuming the use of the ML detector, the argument in Eq. \eqref{ml} can be rewritten as

\begin{equation}
\begin{aligned}
&\operatorname{tr}\left( \mathbf{Y}^H \mathbf{P}_{[\mathbf{X}]}\mathbf{Y}\right) = \operatorname{tr}\left( \mathbf{Y}^H \mathbf{X}\mathbf{X}^H \mathbf{Y}\right) \\ &= \operatorname{tr}\left( \left(\mathbf{X}^H \mathbf{Y}\right)^H  \left(\mathbf{X}^H \mathbf{Y}\right)\right) = \|\mathbf{X}^H \mathbf{Y} \|_F^2,
\end{aligned}
\end{equation} where the computational complexity of evaluating this argument is dominated by the matrix multiplication $\mathbf{X}^H \mathbf{Y}$, which, for arbitrary Grassmannian constellations, results in a total complexity of $\mathcal{O}(LMTN)$ when comparing the $L$ possible codewords. However, for the proposed family of constellations, since each row of $\mathbf{X}$ contains a single nonzero entry whose position is known to the receiver, the multiplication $\mathbf{X}^H \mathbf{Y}$ can be performed by simply adding scaled rows of $\mathbf{Y}$. This reduces the complexity to $\mathcal{O}(LTN)$, yielding an improvement by a factor of $M$.

Nevertheless, the main drawback of ML detection generally arises from the exponential growth of $L$ with the target rate, which is further exacerbated by time-block transmission, yielding $L = 2^{RT}$ for a given rate. However, in the proposed design, the focus on a limited number of points further justifies the use of ML detection. Considering $T = 2M$, for the maximum constellation size of $L = 4M^2$, the resulting computational complexity of $\mathcal{O}(LTN)$ can be expressed as $\mathcal{O}(8M^3N)$, where constellation designs in the literature typically consider $M \leq 5$. Under these conditions, reduced-complexity detectors would offer limited computational gains for the considered family of constellations, while invariably degrading error performance. Finally, note that detection is performed only once every $T$ symbols. Therefore, when compared with non–time-block detection schemes, the computational complexity per symbol period becomes $\mathcal{O}(LN)$, or, for the maximum possible number of $L$, $\mathcal{O}(4M^2N)$.

\section{Numerical Results}\label{res}

In this section, we present numerical results based on SER and BER Monte Carlo simulations to evaluate the error performance of the designed constellations. In all cases, we consider the block-fading channel presented in Section \ref{SM} and the ML detector described in Eq. \eqref{ml}. For this scenario, the SNR is defined in Eq. \eqref{Y} by the term $\rho$. We always consider $T = 2M$ (Section \ref{gsa}), and $2 \leq L \leq 4M^2$, being $L$ a power of two and its maximum restricted by the proposed constellation design (Section \ref{cd}). In Section \ref{HWI}, we also present a preliminary analysis of certain hardware implementation aspects.

\subsection{SER and Achievable Rate Performance}

We begin by analyzing the SER performance of the proposed family of Grassmannian constellations across different scenarios. Fig. \ref{fr1} presents the SER results for a small constellation size, $L = 4$, under different values of $T$, $M$, and $N$. For the considered cases $T = 2$, $T = 4$, and $T = 8$, the corresponding spectral efficiencies are $R = 1$, $R = 0.5$, and $R = 0.25$ bps/Hz, respectively. As expected, for a fixed number of constellation points, increasing the dimensions determined by $(T, M)$ consistently improves the SER performance for the same value of receive antennas, $N$. Furthermore, for each pair $(T, M)$, and consistent with previous theoretical results in the literature (Section \ref{gsa}), the number of receive antennas should be at least $N=\lfloor T/2 \rfloor$ \cite{gt3}. This behavior is confirmed in Fig. \ref{fr1_2} and Fig. \ref{fr1_3}, where a remarkable improvement in performance is observed when $N$ reaches $N = 2$ and $N = 4$, respectively. Nevertheless, as observed across all figures, increasing the number of receive antennas continues to have a remarkable impact on SER performance, even when $N > \lfloor T/2 \rfloor$. Although, as explained in Section \ref{gsa},  when $T < M + N$ the capacity of Grassmannian signaling is outperformed by beta-variance space modulation, this difference becomes significant only when $N \gg T$ \cite{gt5}. This justifies studying the application of Grassmannian signaling with values of $N$ close to $T$, even when $T < M + N$.

\begin{figure*}
\centering
\subfloat[\normalsize $R=1$ $\textrm{bps/Hz}$ \label{fr1_1}]{%
      \includegraphics[width=0.33\linewidth]{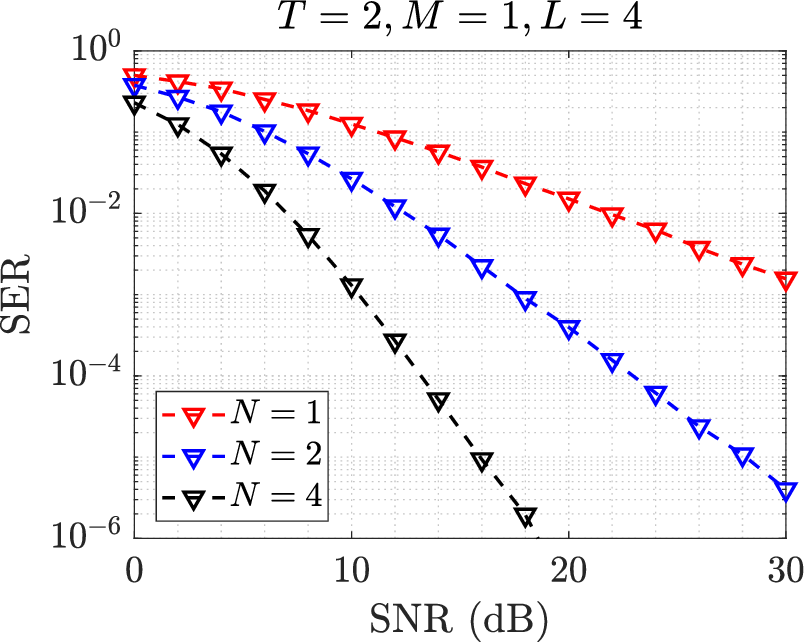}}
      \hfill
\subfloat[\normalsize  $R=0.5$ $\textrm{bps/Hz}$ \label{fr1_2}]{%
      \includegraphics[width=0.33\linewidth]{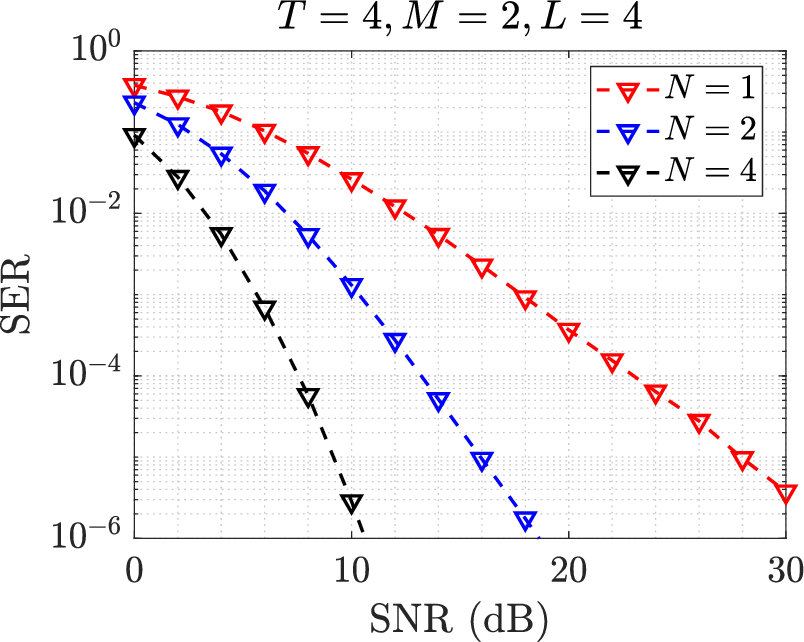}}
        \hfill
\subfloat[\normalsize $R=0.25$ $\textrm{bps/Hz}$  \label{fr1_3}]{%
      \includegraphics[width=0.33\linewidth]{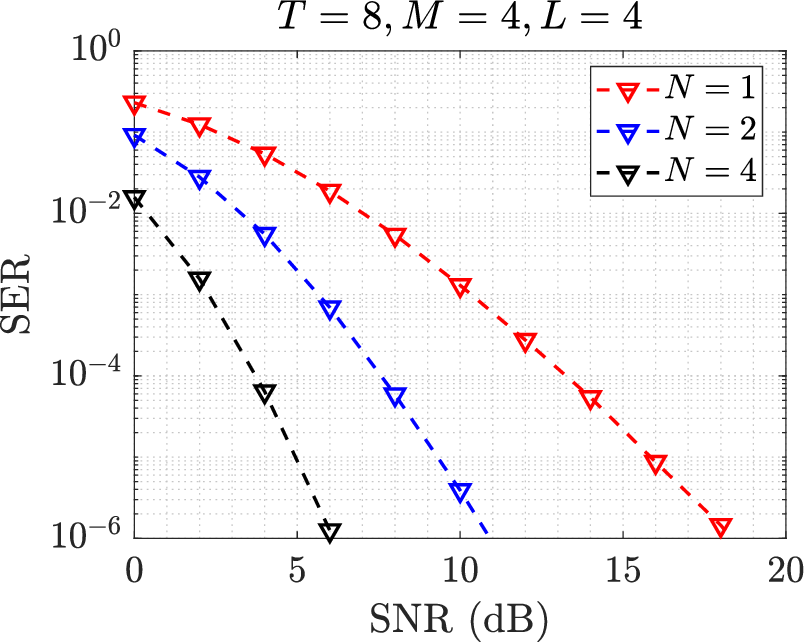}}
\caption{SER results for $L=4$ constellation points and different values of $T$, $M$, and $N$.}
\label{fr1}
\end{figure*}

Fig. \ref{fr2} compares different values of $T$ and $M$ with fixed values of receive antennas, $N = 4$, and spectral efficiency, $R = 0.5$ bps/Hz. Remarkably, increasing the pair $(T, M)$ consistently leads to improved SER performance, even though this implicitly requires increasing the number of constellation points $L$ to maintain the spectral efficiency. This condition ensures that the designed algorithm effectively exploits an increment in the dimensionality of the Grassmann manifold.

\begin{figure}
\centering
\includegraphics[width=0.8\linewidth]{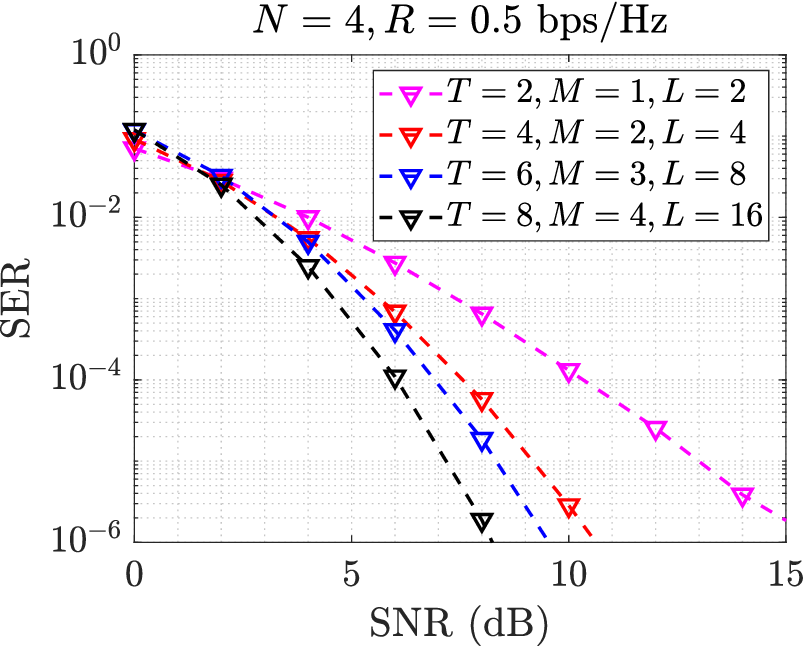}
\caption{SER results across different values of $T$ and $M$ for fixed values of receive antennas and spectral efficiency.}
\label{fr2}
\end{figure}

Finally, Fig. \ref{fr3} presents SER and achievable rate results for all possible values of $L$ with $T = 4$, $M = 2$, and $N = 2$. For $L=2$, $L=4$, $L=8$, and $L=16$, Case (i), Case (ii), Case (iii), and Case (iv) are applied in Algorithm \ref{algc}, respectively. Furthermore, for $L=16$, a comparison is presented between using DP (the default choice) and UB for adjusting geodesic mapping (see Section \ref{ad}). As predicted in Section \ref{ad}, since the constellation points obtained in the Grassmann manifold using both strategies are very similar (Fig. \ref{fsb}), their error performance is nearly identical. Note that, for $L=4$ under Case (ii), the same constellation points are obtained, as shown in Fig. \ref{fsa}. To compute the achievable rate, we perform Monte Carlo simulations based on the mutual information formulas presented in \cite{cubesplit}, adapted to the case of multiple transmit antennas.

 \begin{figure}
\centering
 \subfloat[\normalsize SER \label{fr3a}]{%
       \includegraphics[width=0.8\linewidth]{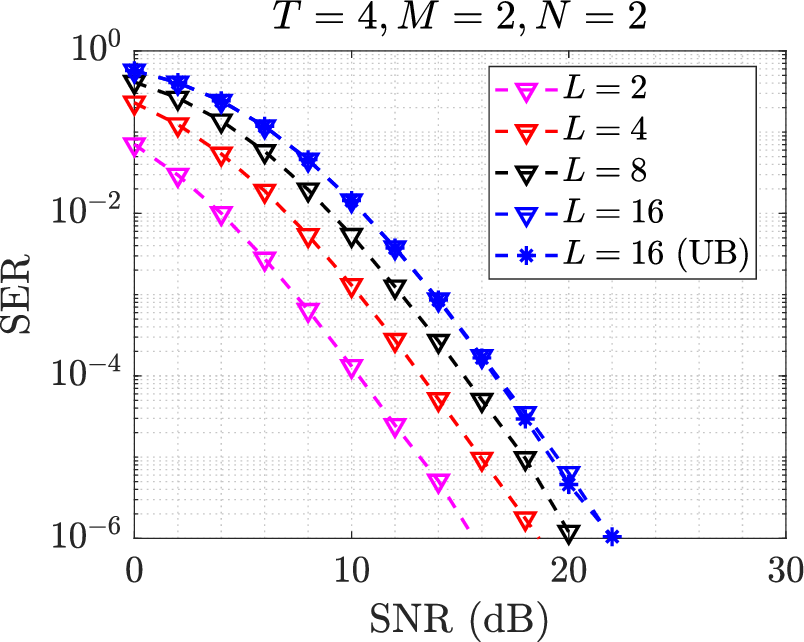}}
  \hfill
  \subfloat[\normalsize Achievable rate \label{fr3b}]{%
        \includegraphics[width=0.8\linewidth]{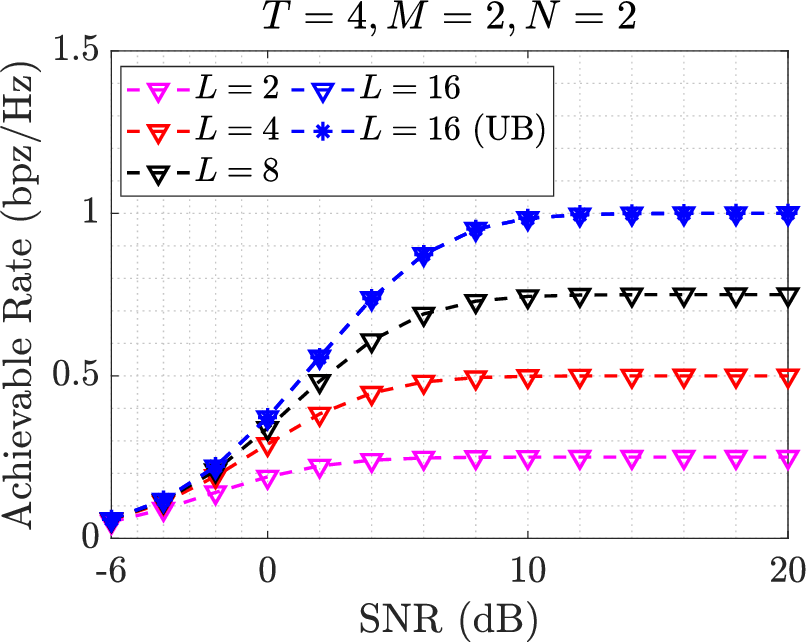}}
  \hfill
  \caption{SER and achievable rate results for all possible values of $L$ and fixed values of $T$, $M$, and $N$. Considering UB instead of DP for geodesic mapping in Algorithm \ref{algc} is denoted as (UB).}
  \label{fr3}
\end{figure}

\subsection{BER versus SER}\label{ubdp}

In Fig. \ref{fr4}, the bit labeling strategy described in Section \ref{bl} is evaluated by comparing SER and BER results across different values of $T$, $M$, $N$, and $L$, always using the maximum possible values of $L$, which for the considered cases are $L = 4$, $L = 16$, and $L = 64$ for $M = 1$, $M = 2$, and $M = 4$, respectively. For $M = 2$ and $M = 4$, the relation $\textrm{BER} \approx \frac{1}{2}\textrm{SER}$ is observed, which corresponds to the expected BER under the assumption that codewords with maximum Hamming distance have a negligible PEP, while all other codewords are equally likely to be mistaken. This validates the design goal of avoiding errors between codewords with maximum Hamming distance. For $M = 1$, we obtain $\textrm{BER} \approx \frac{2}{3},\textrm{SER}$, since in this particular case with $L = 4$, the DP (as well as the chordal and geodesic distances) are nearly identical across all codeword pairs.

\begin{figure}
\centering
\includegraphics[width=0.8\linewidth]{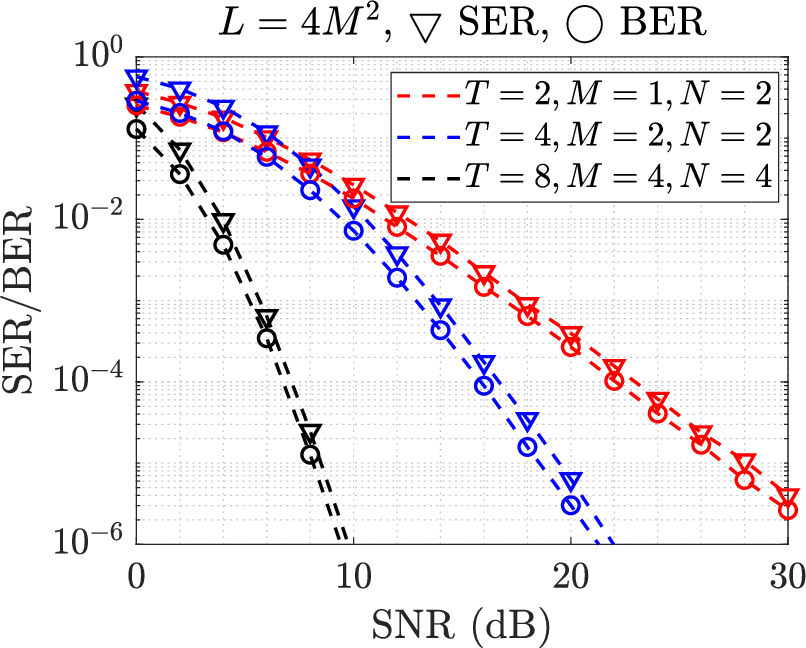}
\caption{Evaluation of bit labeling performance, comparing SER and BER results for equal values of $T$, $M$, $N$, and $L$.}
\label{fr4}
\end{figure}

\subsection{Comparison with Other Approaches}\label{comp}
Finally, we compare the SER performance of our proposed design with other algorithms from the literature. In particular, Fig. \ref{fr5} shows the SER results for the maximum number of points $L = 4M^2$ allowed by our design (GMap.), compared with the DP gradient ascent optimization proposed in \cite{dp} (Opt.), for which the code is publicly available \cite{thesis}. For the optimization case, the default parameters provided by the authors in the code are used, including a maximum of $400$ iterations and a minimum performance improvement threshold of $10^{-5}$, and the best result out of $100$ independent optimization runs is selected. Table \ref{table4} reports the corresponding values of the DP and chordal distance for the constellations obtained with both strategies. For $(T, M) = (2,1)$ with $L = 4$, the performance is identical in both cases. For $(T, M) = (4, 2)$ with $L = 16$, the geodesic mapping approach improves upon the optimization in terms of DP, resulting in slightly better SER performance. Remarkably, for $(T, M) = (8,4)$ with $L = 64$, Case (v) of Algorithm \ref{algc} is applied. In this case, the proposed design does not use DP as the objective function but instead relies on the chordal distance, as explained in Section \ref{ad}. As a result, the obtained DP is $0$, while the chordal distance is excellent, with a minimum value of $\sqrt{M}/\sqrt{2}$ when mapping all geodesics ($L = 64$), where the maximum possible chordal distance between two points is $\sqrt{M} = 2$. Notably, SER results in this case are nearly identical to those obtained with DP optimization. This is because, for the considered values $T = 2M$ with $M = N = 4$, the SER values of interest lie in the low-SNR range (below $10$ dB), where the chordal distance serves as a reliable predictor of the PEP (Section \ref{SEP}). For additional reference, for the particular case $T=4$, $M=2$, and $L=16$, Fig. \ref{fr5} also reports the SER performance of the structured, closed-form constellation obtained by embedding the product of two $M\times M$ unitary matrices into $T\times M$ codewords, following the geometrical design framework of \cite[Sections~IV--V]{geo_di}, denoted by $(\mathbf{A}^k\mathbf{B}^k)$.

In \cite[Fig. 3.4]{thesis}, it can be observed that, using the UB optimization proposed in \cite{ub} with $T = 4$, $M = 2$, $N = 2$, $L = 16$, and SNR$=20$ dB, the obtained SER is about $10^{-5}$, similar to that achieved with the DP optimization in Fig. \ref{fr5}, and slightly worse than our result. In general, results obtained with UB optimization have been shown to slightly improve upon those of DP optimization \cite[Fig. 3.5]{thesis}, and thus, we expect them to be comparable to the results achieved with our approach. Note that the main advantage of the UB optimization presented in \cite{ub} is its applicability to arbitrary values of $T$, $M$, $N$, and $L$, whereas our design is constrained in spectral efficiency by the condition $2 \leq L \leq 4M^2$. On the other hand, for the considered values of $L$, the geodesic mapping approach has the important advantage of requiring only one active antenna per time slot, without degrading SER performance compared to state-of-the-art unstructured Grassmannian constellation designs.

\begin{figure}
\centering
\includegraphics[width=0.8\linewidth]{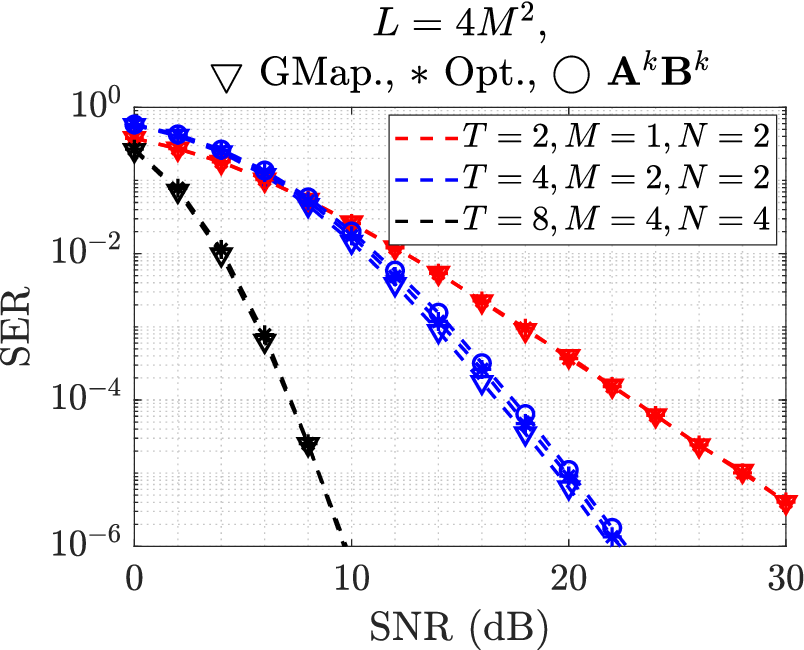}
\caption{SER comparison between the proposed Grassmannian constellation design (GMap.), the DP gradient optimization presented in \cite{dp} (Opt.), and the geometrical design presented in \cite[Sections IV-V]{geo_di} ($\mathbf{A}^k\mathbf{B}^k$) for $T=4$, $M=2$, and $L=16$.}
\label{fr5}
\end{figure}

\begin{table}[t]
\caption{Comparison of DP and Chordal Distance Between the Proposed Algorithm and the Optimization presented in \cite{dp} for $L=4M^2$.}
\centering
\begin{tabular}{|c|c|c|c|c|}
\hline
\multirow{2}{*}{$(T,M)$} & \multicolumn{2}{c|}{\textbf{Proposed}} & \multicolumn{2}{c|}{\textbf{Optimization \cite{dp}}} \\ \cline{2-5}
 & \textbf{$\textrm{DP}\left( \mathcal{X} \right)$} & \textbf{$\textrm{d}_\textrm{c}\left( \mathcal{X} \right)$} & \textbf{$\textrm{DP}\left( \mathcal{X} \right)$} & \textbf{$\textrm{d}_\textrm{c}\left( \mathcal{X} \right)$} \\ \hline \hline

$(2,1)$ & $0.6665$ & $0.8164$ & $0.6665$ & $0.8164$ \\ \hline

$(4,2)$ & $0.1715$ & $0.9102$ & $0.1263$ & $0.8494$ \\ \hline

$(8,4)$ & $0$      & $1.4142$ & $0.0046$ & $1.1467$ \\ \hline

\end{tabular}
\label{table4}
\end{table}

For additional reference, in Fig. \ref{fr6} we compare the obtained BER results for the proposed noncoherent design with those of space-time block coding (STBC) strategies that rely on CSIR but no CSIT \cite{book_stbc}. For $M = N = 2$, the noncoherent design (Noncoh.) uses the maximum spectral efficiency with $T = 4$ and $L = 16$ points, yielding $R = 1$ bps/Hz. For a fair comparison, the same spectral efficiency is selected for the STBC design (noted as CSIR), which in this case is based on the well-known Alamouti scheme with two transmit antennas \cite{alamouti}. Although this design strictly requires only a time block of $T = 2$ if CSIR is already available, under the assumption of a block-fading channel with $T = 4$, reliable CSIR acquisition within each block requires the transmitter to allocate $M$ time slots for transmitting orthogonal pilot symbols from each antenna \cite{mimo_tran}. As a result, within $T = 4$ time slots, only two information symbols are transmitted, requiring the use of a QPSK constellation (instead of BPSK) to achieve $R = 1$ bps/Hz, assuming that pilot symbols have the same average power as information symbols. Furthermore, in this case, the CSIR obtained via maximum-likelihood estimation will be imperfect due to AWGN and the reduced number of pilot symbols, which is accounted for in the Monte Carlo simulations. Under these conditions, the Alamouti scheme shows a similar performance in BER compared to the noncoherent approach. However, note that the Alamouti scheme, applicable only to $M = 2$, offers the best performance among STBC designs that rely on CSIR, due to its ability to simultaneously achieve full rate, full diversity, and symbol orthogonality, which is not possible for $M > 2$ \cite{book_stbc}. For $M = N = 4$, the noncoherent strategy uses $L = 64$ with $T = 8$, yielding $R = 0.75$ bps/Hz. To preserve full diversity and symbol orthogonality, the STBC design adopts a $3/4$ rate scheme, meaning that only three information symbols are transmitted over every four time slots \cite{book_stbc}, \cite{stbc34}. The same conditions for acquiring CSIR apply, resulting in a QPSK constellation achieving a rate of $0.75$ bps/Hz. In this case, the noncoherent strategy shows a consistent improvement in BER performance compared to the coherent approach, with an observed gap of about $1.5$ dB in SNR for low error probabilities. Finally, it should be noted that CSIR, while generally more feasible to obtain than CSIT, inevitably increases system complexity at both the transmitter and receiver, and that these coherent approaches require more than one active antenna per transmission time slot.

\begin{figure}
\centering
\includegraphics[width=0.8\linewidth]{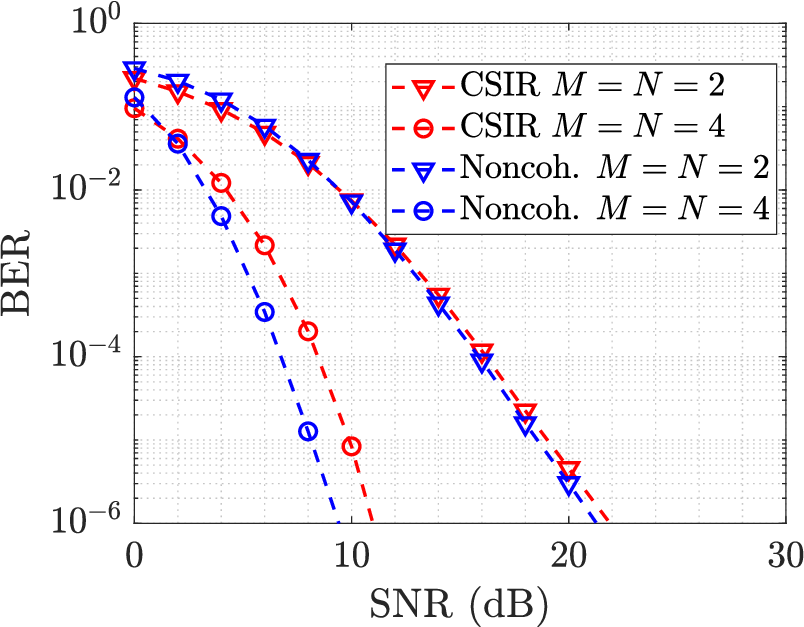}
\caption{BER comparison for equal spectral efficiency between the proposed noncoherent design and STBC strategies relying on CSIR. Spectral efficiencies are $1$ and $0.75$ bps/Hz for $M=N=2$ and $M=N=4$, respectively.}
\label{fr6}
\end{figure}

\section{Hardware Implementation Considerations}\label{HWI}

Finally, we present a preliminary evaluation of the feasibility of the hardware implementation of the proposed scheme. Although an in-depth analysis is beyond the scope of this work, any MIMO system based on a single RF chain and a switching mechanism, as proposed here to fully exploit the presented constellation design, is subject either to a reduction in data rate or to in-band distortion and adjacent-channel interference caused by the truncation of pulse-shaped signals \cite{hwi}.  Naturally, this effect does not appear for $M = 1$, and for $M > 2$ it could be significantly mitigated by using multiple RF chains, as shown in \cite{hwi} (while still using fewer than $M$, thereby preserving the system advantage in this regard). Nevertheless, for the case $M = 2$ (or $M > 2$ with a single RF chain), these effects cannot be avoided and therefore must be analyzed, even if only in a preliminary manner. If spectral-efficiency reduction is to be avoided, the adjacent-channel interference (also referred to as out-of-band distortion) can be mitigated by introducing RF filters after the switching mechanism and before the antennas. Note that this requires multiple filters in exchange for reducing the number of RF chains before the switch (including DACs, mixers, power amplifiers, bias circuits, etc.), with the advantage that filters are passive elements.

Nevertheless, such frequency filtering inevitably causes time spreading and, consequently, inter-symbol interference (ISI). Let us consider a root-raised-cosine (RRC) pulse with roll-off factor $\alpha$, denoted by $g(t)$. Assuming an ideal switch, the pulse is truncated in time to the interval $[-T/2, T/2)$. We can instead intentionally truncate and normalize the signal in the digital domain so that the switch only has the effect of selecting the active antenna. Accordingly, we denote the truncated signal by $\hat{g}(t)$, which satisfies $\int_{-\infty}^{\infty} |\hat g(t)|^2 dt = 1$. It is clear that the system favors higher values of $\alpha$, since they are less distorted by this truncation. For simplicity, we consider the RF filter to have a rectangular frequency response, i.e., an impulse response given by $ h_\beta(t)=\frac{2\beta}{T}\,\operatorname{sinc}\!\left(\frac{2\beta t}{T}\right)$, where $\beta$ is the normalized one-sided bandwidth measured in units of $1/T$, and $\operatorname{sinc}(x)=\frac{\sin(\pi x)}{\pi x}$. We denote the filtered signal by $\hat{g}_f(t)=\hat{g}(t)*h_\beta(t)$, where $*$ denotes convolution.

Filtering results in an energy loss, denoted by $L_\mathrm{P}$, which can be expressed in dB as

\begin{equation}\label{lp}
L_{\mathrm{P}}
=
-10\log_{10}
\left(
\int_{-\infty}^{\infty}
|\hat g_f(t)|^2\,dt
\right).
\end{equation}

The effect of the ISI can be roughly estimated through the residual ISI power normalized to the desired symbol power, denoted by $\rho_{\mathrm{ISI}}$, and computed as $\rho_{\mathrm{ISI}} = 2 \sum_{k=1}^{100} |d_k|^2$, where, for simplicity, only the one hundred closest (most relevant) symbol terms are considered, with $d_k$ being given by

\begin{equation} \label{dk}
d_k
=
\frac{
\int_{-\infty}^{\infty}
\hat g_f(t-kT)\hat g_f^*(t)\,dt
}{
\int_{-\infty}^{\infty}
|\hat g_f(t)|^2\,dt
},
\end{equation} where a matched filter for $\hat{g}_f(t)$ is used at the receiver. By numerically evaluating Eq.~(\ref{lp}) and Eq.~(\ref{dk}) with standard bandwidth $\beta = (1+\alpha)/2$ for $\alpha=0.2,0.35,0.5$ values, we obtain $L_{\mathrm{P}} = 0.66,0.46,0.32$ dB, and $\rho_{\mathrm{ISI}}$ values that correspond to signal-to-interference ratios (SIRs) of $18.26$ dB, $22.72$ dB and $28.17$ dB, respectively. Consequently, the combined effects of filtering-induced power loss and ISI transform the ideal SNR values into lower signal-to-noise-plus-interference ratio (SNIR) values, resulting in corresponding performance degradation, which can be adjusted through the selection of $\alpha$. Additionally, it is clear that the impact of ISI specifically becomes more pronounced at higher SNR values. As design examples, for $M=N=2$ (the most critical case), we consider an SNR of $20$ dB, which already achieves low SER and BER values (below $10^{-5}$) for the maximum supported rate (see Fig. \ref{fr4}). Selecting $\alpha=0.5$ results in an SNIR of $19.11$ dB, with a total degradation of less than $1$ dB. On the other hand, for $M=N=4$ with a single transmit RF chain, an SNR of $10$ dB is sufficient (see Fig. \ref{fr4}), allowing $\alpha$ to be reduced; specifically, $\alpha=0.35$ results in an SNIR of $9.34$ dB, again with a degradation below $1$ dB. Although this is a simplified analysis, it shows that the impact of ISI and power loss from pulse truncation and filtering remains limited, even when $\alpha$ is restricted to commonly used values.

\section{Conclusion}\label{con}

In this work, a structured family of Grassmannian constellations for MIMO noncoherent communications is presented. We begin by reviewing the system model and the metrics associated with the Grassmann manifold that have been previously analyzed in relation to SEP and PEP. Then, we present the proposed algorithm, which is based on computing and mapping geodesic curves on the Grassmann manifold departing from an initial point with a suitable set of initial velocities or tangent vectors. These initial velocities are systematically constructed using the canonical matrix representation of the finite-dimensional Weyl–Heisenberg group, although other sets of vectors satisfying the required properties could also be used. In this manner, we obtain the property that all points or matrices in the unitary space-time constellation contain a single nonzero entry per row, meaning that only one transmit antenna is active per time slot. This approach reduces hardware complexity and implementation cost, as the transmitter can be implemented with a single DAC, RF chain, and power amplifier, regardless of the number of available antennas. Power consumption is also reduced, since only one transmit chain is required along with an RF switch to select the active antenna per time slot. Additionally, the computational cost of ML detection is reduced by a factor of $M$. Considering $T=2M$, and taking at most one point per computed geodesic (except for $L=2$), the design limits the maximum number of points to twice the dimension of the complex Grassmannian (as a real manifold), i.e., $L \leq 4M(2M - M) = 4M^2$. Although this restricts spectral efficiency to the range of $0.25$–$1$ bps/Hz for the combinations of $T$ and $L$ of interest, this range is well suited to many use cases where noncoherent communications are applicable, such as short-packet URLLC communications in mobile scenarios where CSI acquisition is costly or unreliable. Furthermore, this design may be applicable to scenarios where SIMO has previously been considered, such as the uplink in IoT networks, by simply adding more antennas and a switching mechanism to the single-antenna transmitter.

Within the considered number of points $L$, the geodesic mapping can be adjusted to achieve good performance in terms of metrics related to SEP and PEP, while also allowing for a simple bit-labeling strategy that minimizes PEP between codewords with maximum Hamming distance. All of these properties are validated through Monte Carlo simulations for SER and BER. Remarkably, the proposed constellations achieve SER performance comparable to state-of-the-art unstructured designs, while maintaining their unique structure. While it is not always possible to achieve a nonzero DP (also known as a full diversity constellation or code), two factors mitigate its impact on error performance. First, in those cases, excellent chordal distance is achieved instead. Second, this situation arises only for $M \geq 3$, where, considering $N \geq M$, and given the relatively low number of points $L$, the SER values of interest fall within the low-SNR range, in which chordal distance serves as a good metric for reducing PEP. For additional reference, a comparison is presented with standard coherent STBC schemes that rely solely on CSIR. Grassmannian siganling outperforms this designs for $M>2$, especially due to the impact of imperfect CSIR in short coherent intervals, showing similar performance for $M=2$ with the Alamouti schemes. Furthermore, the considered STBC designs require more than one active antenna per time slot. 

Finally, we present a preliminary analysis of the hardware implementation feasibility of the proposed system when based on a single RF chain and a switching mechanism. Although such a system would incur out-of-band distortion and ISI, the former can be mitigated through the use of additional passive RF filters before the antennas (in exchange for reducing the number of RF chains with active elements), while preliminary calculations indicate that the impact of the latter remains limited, since the proposed constellations for $M, N > 1$ achieve low BER and SER values (below $10^{-5}$) for SNR values below $20$ dB even at their maximum supported rates.

To ensure reproducibility of the results, the source code of this work, along with examples of constellations obtained using our design, is available at: https://github.com/alvpr/grassmannian.

\appendices
\section*{Appendix}

\subsection{Theoretical Background} \label{th}
In this section, we present the background that provides the theoretical foundation of the paper, focusing on the elements that are explicitly used in this work. All of this material is standard and can be found, for example, in \cite{ghb} and \cite{cgrass}.

The {\em complex Grassmann manifold} or {\em complex Grassmannian} $\textup{Gr}_{\mathbb{C}}(T,M)$ is defined as the set of all $M$-dimensional complex subspaces $\mathcal{U}$ of $\mathbb{C}^T$. Each such subspace $\mathcal{U}$ can be represented by a matrix $\mathbf{U} \in \mathbb{C}^{T\times M}$ whose $M$ columns form an orthonormal basis of $\mathcal{U}$; we write in this case $\mathcal{U}=\operatorname{span}(\mathbf{U})$. These matrices are elements of the {\em complex Stiefel manifold}
\begin{equation*}
    \textup{St}_{\mathbb{C}}(T,M) \coloneqq \left\{\mathbf{U} \in \mathbb{C}^{T\times M}: \mathbf{U}^H\mathbf{U}=\mathbf{I}_M \right\}.
\end{equation*}
Given $\mathcal{U}$, there is a unique orthogonal projector $\mathbf{P}$ onto $\mathcal{U}$, that is, there exists a unique matrix $\mathbf{P} \in \mathbb{C}^{T\times T}$ such that $\mathbf{P}^H=\mathbf{P}$, $\mathbf{P}^2=\mathbf{P}$, and $\operatorname{range}(\mathbf{P})=\mathcal{U}$, which is given by $\mathbf{P} = \mathbf{U}\mathbf{U}^H$ for any $\mathbf{U} \in \textup{St}_{\mathbb{C}}(T,M)$ such that $\mathcal{U}=\operatorname{span}(\mathbf{U})$. Therefore, for any two matrices $\mathbf{U}_1,\mathbf{U}_2 \in \textup{St}_{\mathbb{C}}(T,M)$, it holds that $\operatorname{span}(\mathbf{U}_1)=\operatorname{span}(\mathbf{U}_2)$ if and only if $\mathbf{U}_2\mathbf{U}_2^H=\mathbf{U}_1\mathbf{U}_1^H$ or, equivalently, if and only if $\mathbf{U}_2=\mathbf{U}_1\mathbf{R}$, where $\mathbf{R}\in \mathbb{C}^{M\times M}$ is an element of the {\em unitary group}
\begin{equation*}
    U(M) \coloneqq \left\{\mathbf{R} \in \mathbb{C}^{M\times M}: \mathbf{R}^H\mathbf{R}=\mathbf{I}_M=\mathbf{R}\mathbf{R}^H \right\}.
\end{equation*}
This means that $\textup{Gr}_{\mathbb{C}}(T,M)$ can be identified with the quotient space $\textup{St}_{\mathbb{C}}(T,M)/U(M)$, i.e., each element $\mathcal{U} = \operatorname{span}(\mathbf{U})$ of $\textup{Gr}_{\mathbb{C}}(T,M)$ uniquely corresponds to the {\em equivalence class}
\begin{equation*}
    [\mathbf{U}] \coloneqq \left\{\mathbf{U}_1 \in \textup{St}_{\mathbb{C}}(T,M): \mathbf{U}_1=\mathbf{U}\mathbf{R}, \ \mathbf{R} \in U(M)\right\},
\end{equation*}
which, in turn, uniquely corresponds to the orthogonal projector $\mathbf{P}_{[\mathbf{U}]} \coloneqq \mathbf{U}\mathbf{U}^H$ onto $\mathcal{U}$. We say that $\mathbf{U} \in \textup{St}_{\mathbb{C}}(T,M)$ is a {\em Stiefel representative} of the subspace $\mathcal{U} = \operatorname{span}(\mathbf{U})$. Throughout this work, we identify $\mathcal{U} \equiv [\mathbf{U}]$ and write the elements of the Grassmannian, referred to as {\em points}, as $[\mathbf{U}] \in \textup{Gr}_{\mathbb{C}}(T,M)$.

$\textup{Gr}_{\mathbb{C}}(T,M)=\textup{St}_{\mathbb{C}}(T,M)/U(M)$ can be endowed with a differential structure that makes it a real compact manifold of dimension $2M(T-M)$. At each $[\mathbf{U}] \in \textup{Gr}_{\mathbb{C}}(T,M)$, the tangent space can be identified with
\begin{equation}
\label{eq:tangent}
    T_{[\mathbf{U}]}\textup{Gr}_{\mathbb{C}}(T,M) = \left\{\mathbf{\Delta} \in \mathbb{C}^{T\times M}: \mathbf{U}^H\mathbf{\Delta}=0\right\},
\end{equation}
where each element $\mathbf{\Delta} \in T_{[\mathbf{U}]}\textup{Gr}_{\mathbb{C}}(T,M)$ is called a {\em tangent vector} to $\textup{Gr}_{\mathbb{C}}(T,M)$ at $[\mathbf{U}]$.

We also endow $\textup{Gr}_{\mathbb{C}}(T,M)$ with the following Riemannian metric:
\begin{equation}
\label{eq:riem_metric}
    g_{[\mathbf{U}]}\left(\mathbf{\Delta}_1,\mathbf{\Delta}_2\right) \coloneqq \mathfrak{Re} \ \operatorname{tr}\left(\mathbf{\Delta}_1^H\mathbf{\Delta}_2\right),
\end{equation}
and given any curve $\gamma: [t_1,t_2] \rightarrow \textup{Gr}_{\mathbb{C}}(T,M)$, we define its {\em $g$-length} as
\begin{equation*}
    L_g(\gamma) \coloneqq \int_{t_1}^{t_2} \sqrt{g_{\gamma(t)}\left(\dot{\gamma}(t),\dot{\gamma}(t) \right)} \ \mathrm{d}t = \int_{t_1}^{t_2} ||\dot{\gamma}(t)||_g \ \mathrm{d}t,
\end{equation*}
where $||\cdot||_g \coloneqq \sqrt{g(\cdot,\cdot)}$ is the {\em $g$-norm} and $\dot{\gamma}(t) = \frac{d}{dt}\gamma(t) \in T_{\gamma(t)}\textup{Gr}_{\mathbb{C}}(T,M)$ denotes the velocity vector of the curve.

With this Riemannian metric, $\textup{Gr}_{\mathbb{C}}(T,M)$ is actually a totally geodesic Riemannian submanifold of the {\em real Grassmann manifold} $\textup{Gr}_{\mathbb{R}}(2T,2M)$ of $2M$-dimensional real subspaces of $\mathbb{R}^{2T}$, endowed with one half times the Frobenius (or Euclidean) metric, as in \cite[Section 3.1]{ghb}. This implies that all the formulas for angles, distances, and geodesics of the real Grassmannian translate exactly the same to the complex case, simply by replacing the transpose with the conjugate transpose. In particular:
\begin{itemize}
    \item The {\em principal angles} $\theta_1,\ldots,\theta_M \in \left[0,\frac{\pi}{2}\right]$ between two points $[\mathbf{U}_1], [\mathbf{U}_2] \in \textup{Gr}_{\mathbb{C}}(T,M)$ are a measure of the ``relative position'' between the two complex subspaces $\mathcal{U}_1=\operatorname{span}(\mathbf{U}_1)$ and $\mathcal{U}_2=\operatorname{span}(\mathbf{U}_2)$. They can be computed as
    \begin{equation}\label{pan}
        \theta_m = \arccos{(\sigma_m)} \in \left[0,\frac{\pi}{2}\right], \quad m = 1,\ldots, M,
    \end{equation}
    where $\sigma_m\in [0,1]$ is the $m$-th largest singular value of $\mathbf{U}_1^H\mathbf{U}_2$. This is independent of the chosen Stiefel representatives $\mathbf{U}_1, \mathbf{U}_2 \in \textup{St}_{\mathbb{C}}(T,M)$. In this work we are most interested in the case where all the principal angles are nonzero, which means that $\mathcal{U}_1 \cap \mathcal{U}_2 = \{0\}$. If one principal angle is $0$, then $\mathcal{U}_1$ and $\mathcal{U}_2$ share a nontrivial direction. If $\theta_m=0$ for all $m$, then $\mathcal{U}_1 = \mathcal{U}_2$. In contrast, if one principal angle is $\frac{\pi}{2}$, there is a direction in $\mathcal{U}_1$ orthogonal to $\mathcal{U}_2$, and vice versa. If $\theta_m=\frac{\pi}{2}$ for all $m$, then $\mathcal{U}_1$ and $\mathcal{U}_2$ are mutually orthogonal: every direction in $\mathcal{U}_1$ is orthogonal to every direction in $\mathcal{U}_2$.
    \item The {\em Riemannian distance} $\textrm{d}_\textrm{g}$ (induced by $g$ in Eq. \eqref{eq:riem_metric}) between two points $[\mathbf{U}_1], [\mathbf{U}_2] \in \textup{Gr}_{\mathbb{C}}(T,M)$ can be computed as
    \begin{equation}
    \label{eq:geod_dist}
        \textrm{d}_\textrm{g}\left([\mathbf{U}_1],[\mathbf{U}_2]\right) = \left(\sum_{m=1}^M \theta_m^2 \right)^{1/2},
    \end{equation}
    where $\{\theta_m\}_{m=1}^M$ are the principal angles between both points. This directly implies that $\textrm{d}_\textrm{g}$ is bounded by
    \begin{equation}
    \label{eq:dist_bound}
        \textrm{d}_\textrm{g}\left([\mathbf{U}_1],[\mathbf{U}_2]\right) \leq \sqrt{M}\frac{\pi}{2}, \quad \forall [\mathbf{U}_1],[\mathbf{U}_2] \in \textup{Gr}_{\mathbb{C}}(T,M).
    \end{equation}
    \item Given a point $[\mathbf{U}] \in \textup{Gr}_{\mathbb{C}}(T,M)$ and a vector $\mathbf{\Delta} \in T_{[\mathbf{U}]}\textup{Gr}_{\mathbb{C}}(T,M)$, the {\em geodesic} $\gamma_{\mathbf{\Delta}}(t)$ (with respect to the Riemannian metric $g$ in Eq. \eqref{eq:riem_metric}) departing at $\gamma_{\mathbf{\Delta}}(0)=[\mathbf{U}]$ with velocity $\dot{\gamma}_{\mathbf{\Delta}}(0)=\mathbf{\Delta}$ is given for any $t \in \mathbb{R}$ by
    \begin{equation}
    \label{eq:geod_eq}
        \gamma_{\mathbf{\Delta}}(t) = \left[\mathbf{U}\mathbf{V}\cos{(t\mathbf{\Sigma})}\mathbf{V}^H+\mathbf{Q}\sin{(t\mathbf{\Sigma})}\mathbf{V}^H\right],
    \end{equation}
    where the cosine and sine functions only apply to the diagonal entries of $t\mathbf{\Sigma}$, and $\mathbf{\Delta} = \mathbf{Q}\mathbf{\Sigma}\mathbf{V}^H$ is the compact singular value decomposition (SVD) of $\mathbf{\Delta}$, with $\mathbf{Q} \in \textup{St}_{\mathbb{C}}(T,M)$, $\mathbf{\Sigma}=\operatorname{diag}(\sigma_1,\ldots,\sigma_M)$, and $\mathbf{V} \in U(M)$.
\end{itemize}

Given two points $[\mathbf{U}_1], [\mathbf{U}_2] \in \textup{Gr}_{\mathbb{C}}(T,M)$, there always exists a geodesic from $[\mathbf{U}_1]$ to $[\mathbf{U}_2]$ whose $g$-length is exactly $\textrm{d}_\textrm{g}\left([\mathbf{U}_1],[\mathbf{U}_2] \right)$. This is why $\textrm{d}_\textrm{g}$ in Eq. \eqref{eq:geod_dist} is also called {\em geodesic distance}. Moreover, if we denote by $\gamma_{\mathbf{\Delta}}|_{[t_1,t_2]}$ the restriction of the geodesic $\gamma_{\mathbf{\Delta}}$ to the interval $[t_1,t_2]$, then the following equivalent properties hold for $t_2-t_1>0$ small enough:
    \begin{itemize}
        \item[(i)] $\gamma_{\mathbf{\Delta}}|_{[t_1,t_2]}$ locally realizes the Riemannian distance: the $g$-length of the geodesic segment $\gamma_{\mathbf{\Delta}}|_{[t_1,t_2]}$ is exactly $\textrm{d}_\textrm{g}\left(\gamma_{\mathbf{\Delta}}(t_1),\gamma_{\mathbf{\Delta}}(t_2) \right)$.
        \item[(ii)] $\gamma_{\mathbf{\Delta}}|_{[t_1,t_2]}$ is locally minimizing: the $g$-length of any curve from $\gamma_{\mathbf{\Delta}}(t_1)$ to $\gamma_{\mathbf{\Delta}}(t_2)$ is greater than or equal to the $g$-length of $\gamma_{\mathbf{\Delta}}|_{[t_1,t_2]}$.
    \end{itemize}
However, these properties do not hold globally in general (for arbitrarily big $t_2-t_1$). Given $[\mathbf{U}] \in \textup{Gr}_{\mathbb{C}}(T,M)$, we define the {\em cut instant} of $\mathbf{\Delta} \in T_{[\mathbf{U}]}\textup{Gr}_{\mathbb{C}}(T,M)$ as
    \begin{equation*}
        \text{Cut}_{[\mathbf{U}]}(\mathbf{\Delta}) \coloneqq \operatorname{sup} \left\{t>0: \gamma_{\mathbf{\Delta}}|_{[0,t]} \text{ is minimizing} \right\},
    \end{equation*}
which can be computed as
    \begin{equation}
    \label{eq:cut}
        \text{Cut}_{[\mathbf{U}]}(\mathbf{\Delta}) = \frac{\pi}{2\sigma_1},
    \end{equation}
where $\sigma_1$ is the largest singular value of $\mathbf{\Delta}$.

When $\mathbf{\Delta}$ is $g$-unit, i.e., $||\mathbf{\Delta}||_g \coloneqq \sqrt{g_{[\mathbf{U}]}(\mathbf{\Delta},\mathbf{\Delta})}=1$, then the $g$-length of $\gamma_{\mathbf{\Delta}}|_{[t_1,t_2]}$ is exactly $t_2-t_1$, for all $t_2>t_1$, and due to Eq. \eqref{eq:dist_bound}, the cut instant is bounded by
\begin{equation*}
    \text{Cut}_{[\mathbf{U}]}(\mathbf{\Delta}) \leq \sqrt{M}\frac{\pi}{2}, \quad \forall g\text{-unit } \mathbf{\Delta} \in T_{[\mathbf{U}]}\textup{Gr}_{\mathbb{C}}(T,M).
\end{equation*}
We say that a $g$-unit vector $\mathbf{\Delta}$ is a {\em diametral vector} if $\text{Cut}_{[\mathbf{U}]}(\mathbf{\Delta}) = \sqrt{M}\frac{\pi}{2}$. This means that the geodesic $\gamma_{\mathbf{\Delta}}$ travels the maximum possible distance in the manifold before properties (i) and (ii) cease to hold.

\subsection{Weyl-Heisenberg basis} \label{wh}
In this section, we present a specific choice of matrices that satisfies all the properties required for our constellation design (Section \ref{cd}). These matrices arise from the so-called displacement operators within the framework of the matrix representation of the finite-dimensional Weyl–Heisenberg group (see, e.g., \cite[Chapter~12]{weyl} and \cite{weyl2}).

Let us denote by $\{\ket{0},\ket{1},\ldots,\ket{M-1}\}$ the canonical orthonormal basis of $\mathbb{C}^M$, i.e.,
\begin{equation*}
    \ket{0} =
    \begin{pmatrix}
    1 \\
    0 \\
    \vdots \\
    0
    \end{pmatrix}, \quad
    \ket{1} =
    \begin{pmatrix}
    0 \\
    1 \\
    \vdots \\
    0
    \end{pmatrix}, \quad \ldots, \quad
    \ket{M-1} =
    \begin{pmatrix}
    0 \\
    0 \\
    \vdots \\
    1
    \end{pmatrix}.
\end{equation*}
We define:
\begin{itemize}
    \item The {\em shift operator} $\mathbf{X} \in U(M)$ by
    \begin{equation*}
        \mathbf{X}\ket{k} = \ket{k+1} \quad (\operatorname{mod} \ k), \quad k=0,\ldots,M-1.
    \end{equation*}
    \item The {\em clock operator} $\mathbf{Z} \in U(M)$ by
    \begin{equation*}
        \mathbf{Z}\ket{k} = \omega^k\ket{k}, \quad k=0,\ldots,M-1,
    \end{equation*}
    where $\omega \coloneqq e^{2\pi i/M}$.
\end{itemize}
Note that $\mathbf{X}^n$ has exactly one nonzero entry in each row and $\mathbf{Z}^n$ is diagonal, for any $n=1,\ldots,M-1$, while $\mathbf{X}^M=\mathbf{Z}^M=\mathbf{I}_M$. The {\em Weyl-Heisenberg matrices} $\mathbf{W}_{m,n} \in U(M)$ are then defined by
\begin{equation*}
    \mathbf{W}_{m,n} \coloneqq \tau^{mn} \mathbf{Z}^m \mathbf{X}^n, \quad m,n=0,\ldots,M-1,
\end{equation*}
where $\tau \coloneqq -e^{i\pi/M}$. These matrices inherit the property of having a single nonzero entry per row, and they satisfy
\begin{equation*}
    \operatorname{tr}\left(\mathbf{W}_{m,n}^H \mathbf{W}_{m',n'} \right) = M \delta_{mm'}\delta_{nn'}.
\end{equation*}
Therefore, the $2M^2$ matrices
\begin{equation}
\label{eq:wh_basis}
    \left\lbrace \tilde{\mathbf{\Delta}}_k \right\rbrace_{k=1}^{2M^2} \coloneqq  \left\lbrace \frac{1}{\sqrt{M}}\mathbf{W}_{m,n}, \frac{i}{\sqrt{M}}\mathbf{W}_{m,n} \right\rbrace_{m,n=0}^{M-1}
\end{equation}
satisfy the condition
\begin{equation*}
    \mathfrak{Re} \ \operatorname{tr}\left(\tilde{\mathbf{\Delta}}_k^H\tilde{\mathbf{\Delta}}_{k'}\right) = \delta_{kk'},
\end{equation*}
and they are called the ($g$-orthonormal) {\em Weyl-Heisenberg basis}.

\subsection{Proof of Theorem \ref{thm:geodesic}} \label{pr1}
If $\mathbf{\Delta} \in T_{[\mathbf{U}]}\textup{Gr}_{\mathbb{C}}(2M,M)$ is in the form of Eq. \eqref{eq:tilde_matrix}, with $\sqrt{M}\tilde{\mathbf{\Delta}} \in U(M)$, then its SVD $\mathbf{\Delta} = \mathbf{Q}\mathbf{\Sigma}\mathbf{V}^H$ is given by
\begin{equation*}
    \mathbf{Q} =
    \begin{pmatrix}
    \mathbf{0}_M \\
    \sqrt{M}\tilde{\mathbf{\Delta}}
    \end{pmatrix}, \quad
    \mathbf{\Sigma} = \frac{1}{\sqrt{M}}\mathbf{I}_M, \quad
    \mathbf{V}=\mathbf{I}_M,
\end{equation*}
and then, it is a straightforward computation to derive Eq. \eqref{eq:geodesic} from the geodesic expression in Eq. \eqref{eq:geod_eq}. Moreover, since $\mathbf{\Delta}^H\mathbf{\Delta}=\tilde{\mathbf{\Delta}}^H\tilde{\mathbf{\Delta}}=\frac{1}{M}\mathbf{I}_M$, clearly $\mathbf{\Delta}$ is $g$-unit with all its singular values equal to $\frac{1}{\sqrt{M}}$. From Eq. \eqref{eq:cut}, we obtain
\begin{equation*}
    \text{Cut}_{[\mathbf{U}]}(\mathbf{\Delta}) = \sqrt{M}\frac{\pi}{2},
\end{equation*}
which means that $\mathbf{\Delta}$ is a diametral vector.

\subsection{Proof of Theorem \ref{thm:principal_angles}} \label{pr2}
Let $\gamma_{\mathbf{\Delta}_1}(t_1)$ and $\gamma_{\mathbf{\Delta}_2}(t_2)$ be two geodesic points satisfying Eq. \eqref{eq:geodesic}, with $t_1, t_2 \in \left(0,\sqrt{M}\frac{\pi}{2}\right)$. This means that $\gamma_{\mathbf{\Delta}_1}$ and $\gamma_{\mathbf{\Delta}_2}$ depart from the same point $[\mathbf{U}]$ with initial velocities $\mathbf{\Delta}_1$ and $\mathbf{\Delta}_2$, respectively, such that $\mathbf{\Lambda}_1 \coloneqq \sqrt{M}\tilde{\mathbf{\Delta}}_1, \mathbf{\Lambda}_2 \coloneqq \sqrt{M}\tilde{\mathbf{\Delta}}_2 \in U(M)$. By a slight abuse of notation, we identify each geodesic point with the Stiefel representative given by Eq. \eqref{eq:geodesic} and write
\begin{equation*}
    \mathbf{\Gamma}(t_1,t_2) \coloneqq \gamma_{\mathbf{\Delta}_1}(t_1)^H\gamma_{\mathbf{\Delta}_2}(t_2).
\end{equation*}
Then, using $\tilde{\mathbf{U}}^H\tilde{\mathbf{U}}=\mathbf{U}^H\mathbf{U}=\mathbf{I}_M$, a straightforward computation shows that
\begin{equation*}
    \begin{split}
    \mathbf{\Gamma}(t_1,t_2) = \cos{\alpha_1}\cos{\alpha_2} \ \mathbf{I}_M + \sin{\alpha_1}\sin{\alpha_2} \ \mathbf{\Lambda}_1^H \mathbf{\Lambda}_2,
    \end{split}
\end{equation*}
where $\alpha_1 \coloneqq \frac{t_1}{\sqrt{M}}$ and $\alpha_2 \coloneqq \frac{t_2}{\sqrt{M}}$. Therefore, $\mathbf{\Gamma}(t_1,t_2)$ is a linear combination of $\mathbf I_M$ and $\mathbf{\Lambda}_1^H \mathbf{\Lambda}_2$. This means, on the one hand, that $\mathbf{\Gamma}(t_1,t_2)$ commutes with $\mathbf{\Lambda}_1^H \mathbf{\Lambda}_2$, so both matrices are simultaneously diagonalizable, i.e., the eigenvalues $\mu_m$ of $\mathbf{\Gamma}(t_1,t_2)$ are given by
\begin{equation}
\label{eq:eigenvalues}
    \mu_m = \cos{\alpha_1}\cos{\alpha_2} + \sin{\alpha_1}\sin{\alpha_2} \ \lambda_m, \quad m=1,\ldots,M,
\end{equation}
where $\lambda_m$ are the eigenvalues of $\mathbf{\Lambda}_1^H \mathbf{\Lambda}_2$. On the other hand, since $\mathbf I_M$ and $\mathbf{\Lambda}_1^H \mathbf{\Lambda}_2$ are two commuting normal matrices (observe that $\mathbf{\Lambda}_1^H \mathbf{\Lambda}_2 \in U(M)$), then $\mathbf{\Gamma}(t_1,t_2)$ is also normal, meaning that its singular values $\sigma_m$ are given by
\begin{equation}
\label{eq:sigma}
    \sigma_m = |\mu_m|, \quad m=1,\ldots,M,
\end{equation}
where $|\cdot|$ denotes the complex absolute value. Now, using that $\mathbf{\Lambda}_1^H \mathbf{\Lambda}_2$ is unitary, we can write its eigenvalues as $\lambda_m = \textrm{e}^{i\phi_m}=\cos{\phi_m}+i\sin{\phi_m}$, where $\phi_m \in [0,2\pi)$. Then, from Eq. \eqref{eq:eigenvalues} we obtain:
\begin{equation}
\label{eq:mu}
\begin{split}
    |\mu_m|^2 & = \left(\cos{\alpha_1}\cos{\alpha_2}+\sin{\alpha_1}\sin{\alpha_2}\cos{\phi_m}\right)^2 \\
    & + \left(\sin{\alpha_1}\sin{\alpha_2}\sin{\phi_m}\right)^2.
\end{split}
\end{equation}
Consequently, for each eigenvalue $\lambda_m$ of $\mathbf{\Lambda}_1^H \mathbf{\Lambda}_2 = M\tilde{\mathbf{\Delta}}_1^H \tilde{\mathbf{\Delta}}_2 = M\mathbf{\Delta}_1^H \mathbf{\Delta}_2$:
\begin{itemize}
    \item If $\lambda_m=1$, then $\phi_m=0$ and
    \begin{equation*}
        \sigma_m=|\cos{\alpha_1}\cos{\alpha_2}+\sin{\alpha_1}\sin{\alpha_2}| = |\cos({\alpha_1-\alpha_2})|.
    \end{equation*}
    \item If $\lambda_m\not=1$, then $\phi_m \not= 0$ (so $\cos{\phi_m}<1$) and
    \begin{equation*}
    \begin{split}
        \sigma_m & \leq |\cos{\alpha_1}\cos{\alpha_2}+\sin{\alpha_1}\sin{\alpha_2}\cos{\phi_m}| \\
        & < |\cos({\alpha_1-\alpha_2})| \leq 1.
    \end{split}
    \end{equation*}
\end{itemize}
Finally, recall that the principal angles $\theta_m$ between $\gamma_{\mathbf{\Delta}_1}(t_1)$ and $\gamma_{\mathbf{\Delta}_2}(t_2)$ are given by $\theta_m = \arccos{(\sigma_m})$ (see Eq. \eqref{pan}), which means that $\theta_m=0$ if and only if $\sigma_m=1$. Taking all this into account, we conclude:
\begin{itemize}
    \item If $t_1 \not= t_2$, then $\alpha_1 \not= \alpha_2$ and $\sigma_m < 1$ for all $m=1,\ldots,M$, i.e., all the principal angles $\theta_m$ are nonzero.
    \item If $t_1=t_2$, then $\sigma_m=1$ if and only if $\lambda_m=1$, i.e., the number of principal angles $\theta_m$ equal to zero coincides with the number of eigenvalues $\lambda_m$ equal to $1$.
\end{itemize}

\subsection{Proof of Corollary \ref{cor:opposite}} \label{pr3}
Following the same notation as in the proof of Theorem \ref{thm:principal_angles} (Appendix\ref{pr2}) with $\mathbf{\Delta}_1 = \mathbf{\Delta}$ and $\mathbf{\Delta}_2 = -\mathbf{\Delta}$, note that all the eigenvalues $\lambda_m=\mathrm{e}^{i\phi_m}$ of $-M\mathbf{\Delta}^H\mathbf{\Delta}=-\mathbf{I}_M$ are equal to $-1$, i.e., $\phi_m=\pi$ for all $m=1,\ldots,M$. Using Eq. \eqref{eq:sigma} and Eq. \eqref{eq:mu} with $\phi_m=\pi$ and $t_1=t_2=\sqrt{M}\frac{\pi}{4}$ ($\alpha_1=\alpha_2=\frac{\pi}{4}$), we obtain:
\begin{equation*}
\begin{split}
    \sigma_m & = |\cos{\alpha_1}\cos{\alpha_2}-\sin{\alpha_1}\sin{\alpha_2}| = |\cos({\alpha_1+\alpha_2})| \\
    & = \cos{\left(\frac{\pi}{2}\right)}=0, \quad m=1,\ldots,M,
\end{split}
\end{equation*}
which means that all the principal angles $\theta_m = \arccos{(\sigma_m)}$ between both geodesic points are equal to $\frac{\pi}{2}$.

\bibliographystyle{ieeetr}
\bibliography{references}

\end{document}